\documentclass[10pt,conference]{IEEEtran}

\makeatletter
\def\ps@headings{%
\def\@oddhead{\mbox{}\scriptsize\rightmark \hfil \thepage}%
\def\@evenhead{\scriptsize\thepage \hfil \leftmark\mbox{}}%
\def\@oddfoot{}%
\def\@evenfoot{}}
\makeatother
\usepackage{amsthm}
\usepackage{amsfonts}
\usepackage{mathrsfs}
\usepackage{amsfonts}
\usepackage{amssymb}
\usepackage{graphicx}
\usepackage{epsfig}
\usepackage{psfrag}
\usepackage{amsmath}
\usepackage{array}
\usepackage{cases}
\usepackage{algorithm2e}
\usepackage{subfigure}
\usepackage{cite,graphicx,amsmath,amssymb,color}
\usepackage{anysize}

\usepackage{algorithmic}
\usepackage{stmaryrd}
\usepackage{multirow}
\usepackage{graphicx,times,amsmath} 
\usepackage{epstopdf}

\marginsize{13mm}{13mm}{19mm}{43mm}

\begin{document}


\title{Forwarding, Caching and Congestion Control in Named Data Networks}
\author{Edmund Yeh, Tracey Ho, Ying Cui, Ran Liu, Derek Leong and Michael Burd}
\author{
  \IEEEauthorblockN{Edmund Yeh\thanks{E.~Yeh gratefully acknowledges support from the National Science Foundation grant CNS-1423250 and a Cisco Systems research grant. Y. Cui gratefully acknowledges support from the National Science Foundation of China grant 61401272.}}
  \IEEEauthorblockA{Northeastern University\\
  eyeh@ece.neu.edu}\\
  \IEEEauthorblockN{Ran Liu}
  \IEEEauthorblockA{Northeastern University\\
  liu.ran1@husky.neu.edu}

  \and
  \IEEEauthorblockN{Tracey Ho}
  \IEEEauthorblockA{Speedy Packets Inc.\\
  tracey@speedypackets.com}\\
     \IEEEauthorblockN{Michael Burd}
   \IEEEauthorblockA{Honeywell FM\&T\\
   burdmi@gmail.com}

  \and
   \IEEEauthorblockN{Ying Cui}
   \IEEEauthorblockA{Shanghai Jiao Tong University\\
   cuiying@sjtu.edu.cn}\\
   \IEEEauthorblockN{Derek Leong}
  \IEEEauthorblockA{Inst. for Infocomm Research\\
  dleong@i2r.a-star.edu.sg}
}

\maketitle

\newtheorem{Thm}{Theorem}
\newtheorem{Alg}{Algorithm}
\newtheorem{Def}{Definition}
\newtheorem{Rem}{Remark}

\allowdisplaybreaks

\begin{abstract}
Emerging information-centric networking architectures seek to optimally utilize both bandwidth and storage for efficient content distribution.   This highlights the need for joint design of traffic engineering and caching strategies, in order to optimize network performance in view of both current traffic loads and future traffic demands.   We present a systematic framework for joint dynamic interest request forwarding and dynamic cache placement and eviction, within the context of the Named Data Networking (NDN) architecture.  The framework employs a virtual control plane which operates on the user demand rate for data objects in the network, and an {actual} plane which handles Interest Packets and Data Packets.  We develop distributed algorithms within the virtual plane to achieve network load balancing through dynamic forwarding and caching, thereby maximizing the user demand rate that the NDN network can satisfy.   Next, we show that congestion control can be optimally combined with forwarding and caching within this framework to maximize user utilities subject to network stability.
Numerical experiments within a number of network settings demonstrate the superior performance of the resulting algorithms for the actual plane in terms of high user utilities, low user delay, and high rate of cache hits.

\end{abstract}





\section{Introduction}

Emerging information-centric networking (ICN) architectures are currently changing the landscape of network research.
In particular, Named data networking (NDN)~\cite{Zhang10}, or content-centric networking (CCN)\cite{Jacobson},
is a proposed network
architecture for the Internet that replaces the traditional client-server model of communications with
one based on the identity of data or content.  This abstraction more accurately reflects how the
Internet is primarily used today: instead of being concerned about communicating with specific
nodes, end users are mainly interested in obtaining the data they want.  The NDN architecture offers
a number of important advantages in decreasing network congestion and delays, and in enhancing
network performance in dynamic, intermittent, and unreliable mobile wireless environments~\cite{Zhang10}.

Content delivery in NDN is accomplished using two types of packets, and specific data structures in nodes.  Communication is initiated by the data consumer or requester.  To receive data, the requester sends out an {\em Interest Packet}, which carries the (hierarchically structured) name of the desired data (e.g. {\tt /northeastern/videos/WidgetA.mpg/1}).  The Interest Packet is forwarded by looking up the data name in the {\em Forwarding Information Base (FIB)} at each router the Interest Packet traverses, along routes determined by a name-based routing protocol.
The FIB tells the router to which neighbor node(s) to transmit each Interest Packet.   Each router maintains a {\em Pending Interest Table (PIT)}, which records all Interest Packets currently awaiting matching data.  Each PIT entry contains the name of the interest and the set of node interfaces from which the Interest Packets for the same name arrived.  When multiple interests for the same name are received, only the first is sent toward the data source.  When a node receives an interest that it can fulfill with matching data, it creates a {\em Data Packet} containing the data name, the data content, together with a signature by the producer's key.  The Data Packet follows in reverse the path taken by the corresponding Interest Packet, as recorded by the PIT state at each router traversed.  When the Data Packet arrives at a router, the router locates the matching PIT entry, transmits the data on all interfaces listed in the PIT entry, and then removes the PIT entry.  The router may optionally cache a copy of the received Data Packet in its local {\em Content Store}, in order to satisfy possible future requests.  Consequently, a request for a data object can be fulfilled not only by the content source but also by any node with a copy of that object in its cache~\cite{Zhang10}.

Assuming the prevalence of caches, the usual approaches for forwarding and caching may no longer be effective for
ICN architectures such as NDN.  Instead, these architectures seek to optimally utilize both bandwidth and storage for efficient content distribution.   This highlights the need for joint design of traffic engineering and caching strategies, in order to optimize network performance in view of both current traffic loads and future traffic demands.
Unlike many existing works
on centralized algorithms for static caching, our goal is to develop distributed,
dynamic algorithms that can address caching and forwarding under changing content,
user demands, and network conditions.

To address this fundamental problem, we introduce the {\em VIP framework} for the design of high performing NDN networks.  The VIP framework relies on the new metric of {\em Virtual Interest Packets} (VIPs), which captures the measured demand for the respective data objects in the network.
The central idea of the VIP framework is to employ a \emph{virtual} control plane which operates on VIPs, and an \emph{actual} plane which handles Interest Packets and Data Packets.  Within the virtual plane, we develop distributed control algorithms operating on VIPs, aimed at yielding desirable performance in terms of network metrics of concern.  The flow rates and queue lengths of the VIPs resulting from the control algorithm in the virtual plane are then used to specify the forwarding and caching policies in the actual plane.

The general VIP framework allows for a large class of control and optimization algorithms operating on VIPs in the virtual plane, as well as a large class of mappings which use the VIP flow rates and queue lengths from the virtual plane to specify forwarding and caching in the actual plane.  Thus, the VIP framework presents a powerful paradigm for designing efficient NDN-based networks with different properties and trade-offs.  In order to illustrate the utility of the VIP framework, we present two particular instantiations of the framework.  The first instantiation consists of a distributed forwarding and caching policy in the virtual plane which achieves effective load balancing and adaptively maximizes the throughput of VIPs, thereby maximizing the user demand rate for data objects satisfied by the NDN network.  The second instantiation consists of distributed algorithms which achieves not only load balancing but also stable caching configurations.  Next, we show that congestion control can be naturally combined with forwarding and caching within the VIP framework to
maximize user utilities subject to network stability.  Experimental results show that the latter set of algorithms have superior performance in terms of high user utilities, low user delay and high rate of cache hits, relative to several baseline congestion control, forwarding and caching policies.

We begin with a formal description of the network model in Section~\ref{sec:model}, and discuss the VIP framework in Section~\ref{sec:vipframe}.
We present two instantiations of the VIP framework in Sections~\ref{sec:forwarding-caching-VIP} and~\ref{sec:forwarding-caching-stable}.  Congestion control is discussed in~\ref{sec:congestion-VIP}.  The performance of the proposed forwarding and caching policies is numerically evaluated in comparison with several baseline routing and caching policies using simulations in Section~\ref{sec:simulations}.

Although there is now a rapidly growing literature in information centric networking, the problem of optimal joint forwarding and caching for content-oriented networks remains open.
In \cite{PotentialRouting2012}, a potential-based forwarding scheme with random caching is proposed for ICNs.  A simple heuristically defined measure (called potential value) is introduced for each node.  A content source or caching node has the lowest potential and the potential value of a node increases with its distance to the content source or caching node.  Potential-based forwarding guides Interest Packets from the requester toward the corresponding content source or caching node.
As the Data Packet travels on the reverse path, one node on the path is randomly selected as a new caching node.  The results in~\cite{PotentialRouting2012} are heuristic in the sense that it remains unknown how to guarantee good performance by choosing proper potential values.
In \cite{CachingTM2011:Ying},  the authors consider one-hop routing and caching in a content distribution network (CDN) setting.
Throughput-optimal one-hop routing and caching are proposed to support the maximum number of requests.
Given the simple switch topology, however, routing is reduced to cache node selection.  Throughput-optimal caching and routing in  multi-hop networks remains an open problem.
In \cite{TECC2012}, the authors consider single-path routing and caching to minimize link utilization for a general multi-hop content-oriented network, using primal-dual decomposition within a flow model.  Here, it is assumed that the path between any two nodes is predetermined.  Thus, routing design reduces to cache node selection~\cite{TECC2012}.
The benefits of selective caching based on the concept of betweenness centrality, relative to ubiquitous caching, are shown in~\cite{Chai:2012:CLM:2342042.2342046}.
Cooperative caching within ICNs has been investigated in~\cite{Age-based:6193504}, where an age-based caching scheme is proposed.  These proposed cooperative caching schemes have been heuristically designed, and have not been jointly optimized with forwarding strategies.  Finally, adaptive multipath forwarding in NDN has been examined in~\cite{Yi:2012:AFN:2317307.2317319}, but has not been jointly optimized with caching strategies.

\section{Network Model}

\label{sec:model}

{Consider a connected multi-hop (wireline) network modeled by a directed graph $\mathcal G=(\mathcal N, \mathcal L)$, where $\mathcal N$ and $\mathcal L$ denote the sets of $N$ nodes and $L$ directed links, respectively.  Assume that $(b,a) \in {\cal L}$ whenever $(a,b) \in {\cal L}$.  Let $C_{ab} > 0$ be the transmission capacity (in bits/second) of link $(a,b) \in {\cal L}$.  Let $L_n$ be the cache size (in bits) at node $n \in {\cal N}$ ($L_n$ can be zero).


Assume that content in the network are identified as {\em data objects}, with the object identifiers determined by an appropriate level within the hierarchical naming structure.  These identifiers may arise naturally from the application, and are determined in part by the amount of control state that the network is able to maintain.  Each data object (e.g.~{\tt /northeastern/videos/WidgetA.mpg}) consists of a sequence of {\em data chunks} (e.g.~{\tt /northeastern/videos/WidgetA.mpg/1}).    We assume that any data object is demarcated by a {\em starting chunk} and an {\em ending chunk}.  Content delivery in NDN operates at the level of data chunks.  That is, each Interest Packet requests a particular data chunk, and a matching Data Packet consists of the requested data chunk, the data chunk name, and a signature.    A request for a data object consists of a sequence of Interest Packets which request all the data chunks of the object, where the sequence starts with the Interest Packet requesting the starting chunk, and ends with the Interest Packet requesting the ending chunk.\footnote{The data chunks in between the starting and ending chunks can be requested in any order.}
In the VIP framework which we introduce below, distributed control algorithms are developed in a virtual control plane operating at the data object level, while forwarding of Interest Packets and caching of Data Packets in the actual plane operate at the data chunk level.

We will operate our forwarding and caching algorithms over a set ${\cal K}$ of $K$ data objects in the network.  As mentioned above, ${\cal K}$ may be determined by the amount of control state that the network is able to maintain.  Since the data object popularity distribution evolves at a relatively slow time scale compared to the caching and forwarding, one approach is to let ${\cal K}$ include the set of the most popular data objects in the network, which is typically responsible for most of the network congestion.\footnote{The less popular data objects not in ${\cal K}$ may be distributed using simple techniques such as shortest-path forwarding with little or no caching.}
For simplicity, we assume that
all data objects have the same size $z$ (in bits).   The results in the paper can be
extended to the more general case where object sizes differ.   We consider
the scenario where  $L_n < Kz$ for all $n \in
{\cal N}$.  Thus, no node can cache all data objects.

For each data object $k \in {\cal K}$, assume
that there is a unique node $src(k) \in {\cal N}$ which serves as the
content source for the object.    Interest Packets for chunks of a given data object
can enter the network at any node, and exit the network upon being satisfied
by a matching Data Packet at the content source for the object, or at the
nodes which decide to cache the object.   For convenience, we assume that
the content sources are fixed, while the caching points may vary in time.


Assume that routing (topology discovery and data reachability) has already been accomplished in the network, so that the FIBs have been populated for the various data objects.
Upon the arrival of an Interest Packet at an NDN node, the following sequence of events happen.  First, the node checks its Content Store (CS) to see if the requested data object chunk is locally cached.  If it is, then the Interest Packet is satisfied locally, and a Data Packet containing a copy of the data object chunk is sent on the reverse path.  If not, the node checks its PIT to see if an Interest Packet requesting the same data object chunk has already been forwarded.  If so, the new Interest Packet (interest, for short) is suppressed while the incoming interface associated with the new interest is added to the PIT.  Otherwise, the node checks the FIB to see to what node(s) the interest can be forwarded, and chooses a subset of those nodes for forwarding the interest.   Next, we focus on Data Packets.  Upon receiving a Data Packet, a node needs to determine whether to make a copy of the Data Packet and cache the copy or not.
Clearly, policies for the forwarding of Interest Packets and the caching of Data Packets are of central importance in the NDN architecture.  Thus far, the design of the strategy layer for NDN remains largely unspecified.  Moreover, in the current CCN implementation, a Data Packet is cached at every node on the reverse path.  This, however, may not be possible or desirable when cache space is limited.

We shall focus on the problem of finding dynamic forwarding and caching policies which exhibit superior performance in terms of metrics such as the total number of data object requests satisfied (i.e., all corresponding Data Packets are received by the requesting node), the delay in satisfying Interest Packets, and cache hit rates.  We propose a VIP framework to solve this problem, as described in the next section.



\section{Virtual Interest Packets and\\ the VIP Framework}

\label{sec:vipframe}

\begin{figure}[t]
\begin{center}
\includegraphics[width=80mm,height=25mm]{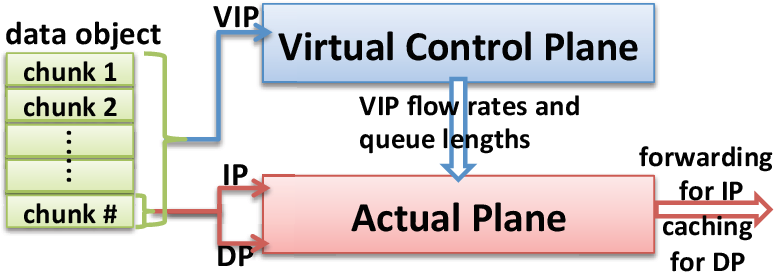}
\caption{VIP framework.  IP (DP) stands for Interest Packet (Data Packet).}\label{fig:planes}
\end{center}
\end{figure}

The VIP framework for joint dynamic forwarding and caching relies on the essential new metric of {\em virtual interest packets} (VIPs), which are generated as follows.   As illustrated in Figure~\ref{fig:planes}, for each request for data object $k \in {\cal K}$ entering the network, a corresponding VIP for object $k \in {\cal K}$ is generated.\footnote{More generally, VIPs can be generated at a rate proportional to that of the corresponding data object requests, which can in some cases improve the convergence speed of the proposed algorithms.}
The {VIPs capture the {\em measured demand} for the respective data objects} in the network, and represent content popularity which is empirically measured, rather than being given a priori.  Specifically, the VIP count for a data object in a given part of the network represents the {\em local} level of interest in the data object, as determined by network topology and user demand.

The VIP framework employs a \emph{virtual} control plane which operates on VIPs {\em at the data object level}, and an \emph{actual} plane which handles Interest Packets and Data Packets {\em at the data chunk level}.   This design has two motivations.  First, this approach reduces the implementation complexity of the VIP algorithm in the virtual plane considerably (as compared with operating on data chunks in the virtual plane).
Second, as shown in Section 4.2 below, this approach leads to a desirable implementation which forwards all the Interest Packets for the same ongoing request for a data object on the same path, and which caches the entire data object (consisting of all data chunks) at a caching node (as opposed to caching different chunks of the same data object at different nodes).   At the same time, the approach also allows Interest Packets for non-overlapping requests for the same data object to be forwarded on different paths, thus making multi-path forwarding of object requests possible.\footnote{In principle, the VIP algorithm in the virtual plane can be applied at the chunk level (corresponding to the case where there is only one chunk in each data object).  In this case, the virtual and actual planes operate at the same granularity.  On the other hand, the complexity of implementing the algorithm in the virtual plane would be much larger.}

Within the virtual plane, we develop distributed control algorithms operating on VIPs, aimed at yielding desirable performance in terms of network metrics of concern.  The flow rates and queue lengths of the VIPs resulting from the control algorithm in the virtual plane are then used to specify the forwarding and caching policies in the actual plane (see Figure~\ref{fig:planes}).
A key insight here is that control algorithms operating in the virtual plane can take advantage of local information on network demand (as represented by the VIP counts), which is unavailable in the actual plane due to interest collapsing and suppression.

In order to illustrate the utility of the VIP framework, we present two particular instantiations of the framework in Sections~\ref{sec:forwarding-caching-VIP} and~\ref{sec:forwarding-caching-stable}.   For both instantiations, the following  hold.  First,  the VIP count is used as the common metric for enabling both the distributed forwarding and distributed caching algorithms in the virtual  and actual control planes.  Second, the forwarding strategy in the virtual plane achieves load balancing through the application of the backpressure algorithm~\cite{Tassiulas-Ephremides:1992-AC} to the VIP queue state.  Finally, one caching algorithm determines the caching locations and cache replacement policy for both the virtual and actual planes.  The two instantiations differ in the manner in which they use the VIP count to determine caching actions.

\subsection{VIP Dynamics}

We now specify the dynamics of the VIPs within the virtual plane.
Consider time slots of length 1 (without loss of generality) indexed by $t = 1, 2, \ldots$.  Specifically, time slot $t$ refers to the time interval $[t,t+1)$.
Within the virtual plane, each node $n\in {\cal N}$ maintains a separate VIP queue  for each data object $k \in {\cal K}$.
Note that no data is contained in these VIPs.  Thus, the VIP queue size for each node $n$ and data object $k$ at the beginning of slot $t$ (i.e., at time $t$) is  represented by a {\em counter} $V_n^k(t)$.\footnote{We assume that VIPs can be quantified as a real number.  This is reasonable when the VIP counts are large.}  Initially, all VIP counters are set to 0, i.e., $V_n^k(1)=0$.
As VIPs are created along with data object requests, the counters for the corresponding data object are incremented accordingly at the entry nodes.  After being forwarded through the network (in the virtual plane), the VIPs for object $k$ are removed at the content source $src(k)$, and at nodes that have cached object $k$.
That is, the content source and the caching nodes are the {\em sinks} for the VIPs.   Physically, the VIP count can be interpreted as a {\em potential}.  For any  data object, there is a downward ``gradient" from entry points of the data object requests to the content source and caching nodes.

An exogenous request for data object $k$ is considered to have arrived at node $n$ if the Interest Packet requesting the starting chunk of data object $k$ has arrived at node $n$.
Let $A^k_n(t)$ be the number of exogenous data object request arrivals at node $n$ for object $k$ during slot $t$ (i.e., over the time interval $[t,t+1)$).\footnote{We think of a node $n$ as a point of aggregation which combines many network users.  While a single user may request a given data object only once, an aggregation point is likely to submit many requests for a given data object over time.}
For every arriving request for data object $k$ at node $n$, a corresponding VIP for object $k$ is generated at $n$ ($V^k_n(t)$ incremented by 1).\footnote{For the general case where object sizes differ, $V^k_n(t)$  is incremented by the object size $z_k$ for every arriving request for object $k$.}
The long-term exogenous VIP arrival rate at node $n$ for object $k$ is
$ \lambda_n^k\triangleq \mathbb \lim_{t \rightarrow \infty} \frac{1}{t} \sum_{\tau = 1}^t A^k_n(\tau).$

Let $\mu_{ab}^k(t)  \geq 0$ be the allocated transmission rate of VIPs for data object $k$ over link $(a,b)$ during time slot $t$.  Note that
at each time $t$ and for each object $k$, a single message between node $a$ and node $b$ can summarize all the VIP transmissions during that time slot.

In the virtual plane, we assume that at each time $t$, each node $n \in {\cal N}$ can gain access to any data object $k \in {\cal K}$ for which there is interest at $n$, and potentially cache the object locally.
Let $s_n^{ k}(t) \in \{0,1\}$ represent the caching state for object $k$ at node $n$ during slot $t$, where $s_n^{k}(t)=1$ if object $k$  is cached at node $n$ during slot $t$, and $s_n^{k}(t)=0$ otherwise.  Now note that even if $s_n^{k}(t)=1$, the content store at node $n$ can satisfy only a limited number of VIPs during one time slot.  This is because there is a maximum rate $r_n$ (in objects per slot) at which node $n$ can produce copies of cached object $k$.\footnote{The maximum rate $r_n$ may reflect the I/O rate of the storage disk. Since it is assumed that all data objects have the same length, it is also assumed that the maximum rate $r_n$ is the same for all data objects.}

The time evolution of the VIP count at node $n$ for object $k$ is as follows:

\begin{align}
& V^k_n(t+1) \leq \nonumber\\
&   \left(
\left(V^k_n(t)-\sum_{b\in \mathcal N}\mu^{k}_{nb}(t)\right)^+ +A^k_n(t)
+\sum_{a\in \mathcal N}\mu^{k}_{an}(t)- r_n s_n^{k}(t)\right)^+
\label{eqn:queue_dyn}
\end{align}
where $(x)^+ \triangleq \max(x,0)$.  Furthermore,
$V^k_n(t) = 0$ for all $t \geq 1$ if $n=src(k)$.

From~\eqref{eqn:queue_dyn}, it can be seen that the VIPs for data object $k$ at node $n$ at the beginning of slot $t$ are transmitted during slot $t$ at the rate $\sum_{b\in \mathcal N}\mu^{k}_{nb}(t)$.  The remaining VIPs $(V^k_n(t)-\sum_{b\in \mathcal N}\mu^{k}_{nb}(t))^+$, as well as the exogenous and endogenous VIP arrivals during slot $t$, are reduced by $r_n$ at the end of slot $t$ if object $k$ is cached at node $n$ in slot $t$ ($s_n^k(t) = 1$). The VIPs still remaining are then transmitted during the next slot $t+1$.
Note that~\eqref{eqn:queue_dyn} is an inequality because the actual number of VIPs for object $k$
arriving to node $n$ during slot $t$ may be less than $\sum_{a\in \mathcal N}\mu^{k}_{an}(t)$ if the neighboring nodes have little or
no VIPs of object $k$ to transmit.

\section{Throughput Optimal VIP Control}


\label{sec:forwarding-caching-VIP}

In this section, we describe an instantiation of the VIP framework in which the VIP count is used as a common metric for enabling both the distributed forwarding and distributed caching algorithms in the virtual and actual control planes.  The forwarding strategy within the virtual plane is given by the application of the backpressure algorithm~\cite{Tassiulas-Ephremides:1992-AC} to the VIP queue state.
Note that while the backpressure algorithm has been used for routing in conventional source-destination-based networks, its use for forwarding in ICNs appears for the first time in this paper.  Furthermore, backpressure forwarding is being used in the virtual plane rather than in the actual plane, where interest collapsing and suppression make the application of the algorithm impractical.

The caching strategy is given by the solution of a max-weight problem involving the VIP queue length.  The VIP flow rates and queue lengths are then used to specify forwarding and caching strategies in the actual plane, which handles Interest Packets and Data Packets.  We show that the joint distributed forwarding and caching strategy adaptively maximizes the throughput of VIPs, thereby maximizing the user demand rate for data objects satisfied by the network.

We now describe the joint forwarding and caching algorithm for VIPs in the virtual control plane.
\begin{Alg}  At the beginning of each time slot $t$, observe the VIP counts $(V^k_n(t))_{k\in \mathcal K, n \in \mathcal N}$ and perform forwarding and caching in the virtual plane as follows.

\textbf{Forwarding}: For each data object $k \in {\cal K}$ and each link $(a,b)\in \mathcal
L^k$,  choose
\begin{align}
\mu^{k}_{ab}(t)
=&\begin{cases} C_{ba}/z,
&  W^*_{ab}(t)>0\  \text{and}\ k=k^*_{ab}(t)\\
0, &\text{otherwise}
\end{cases}\label{eqn:forwarding-VIP}
\end{align}
where
\begin{align}
& W_{ab}^{k}(t) \triangleq V^k_a(t) - V^k_b(t),
\label{eqn:weight}\\
&k^*_{ab}(t) \triangleq \arg
\max_{\{k: (a,b)\in \mathcal L^{k}\}} W_{ab}^{k}(t), \nonumber \\
& W^*_{ab}(t)
\triangleq \left(W_{ab}^{k^*_{ab}(t)}(t)\right)^+.\nonumber
\end{align}
Here,  $\mathcal L^k$ is the set of links which are allowed to transmit the VIPs of object $k$,
$W_{ab}^{k}(t)$ is the backpressure weight of object $k$ on link $(a,b)$ at time $t$, and
$k^*_{ab}(t)$ is the data object which maximizes the backpressure weight on link $(a,b)$ at time $t$.

\textbf{Caching}:  At each node $n \in \mathcal N$, choose $\{s^k_n(t)\}$ to
\begin{equation}
\text{maximize} \  \sum_{k\in \mathcal K}V^k_n(t) s^k_n \quad
\text{subject~to} \ \sum_{k\in \mathcal K} s^k_n\leq L_n/z
\label{eqn:knapsack-VIP}
\end{equation}

Based on the forwarding and caching in \eqref{eqn:forwarding-VIP} and \eqref{eqn:knapsack-VIP}, the VIP count is updated according to \eqref{eqn:queue_dyn}.
\label{Alg:VIP}
\end{Alg}

At each time $t$ and for each link $(a,b)$, backpressure-based forwarding algorithm allocates the entire normalized ``reverse" link capacity $C_{ba}/z$ to transmit the VIPs for the data object $k^*_{ab}(t)$ which maximizes the VIP queue difference $W_{ab}^{k}(t) $ in \eqref{eqn:weight}.  Backpressure forwarding maximally balances out the VIP counts, and therefore the demand for data objects in the network, thereby minimizing the probability of demand building up in any one part of the network and causing congestion.

The caching strategy is given by the optimal solution to the max-weight knapsack problem in \eqref{eqn:knapsack-VIP}, which can be solved optimally in a greedy manner as follows.  For each $n \in {\cal N}$, let $(k_1,k_2,\ldots, k_K)$ be a permutation of $(1, 2, \ldots, K)$ such that $V^{k_1}_n(t)  \geq  V^{k_2}_n(t) \geq \cdots \geq V^{k_K}_n(t)$.  Let $i_ n = \lfloor L_n/z \rfloor$.  Then for each $n \in {\cal N}$, choose
\begin{align}
s_n^{k}(t)
=&\begin{cases} 1,
&  k\in \mathcal \{k_1,\cdots, k_{i_n}\}\\
0, &\text{otherwise}
\end{cases}\label{eqn:greedy-VIP}
\end{align}
Thus, the objects with the highest VIP counts (the highest local popularity) are cached.

It is important to note that both the backpressure-based forwarding algorithm and the max-weight caching algorithm are {\em distributed.}  To implement the forwarding algorithm, each node must exchange its VIP queue state with only its neighbors.  The implementation of the caching algorithm is local once the updated VIP queue state has been obtained.

To characterize the implementation complexity of Algorithm 1, we note that both the computational and communication complexity of the back pressure forwarding algorithm per time slot is $O(N^2 K)$, where the bound can be improved to $O(NDK)$ if $D$ is the maximum node degree in the network.  Assuming fixed cache sizes, the computational complexity of the caching algorithm per time slot can be found to be $O(NK)$.

In the following section, we show that the forwarding and caching strategy described in Algorithm~\ref{Alg:VIP} is {\em throughput optimal} within the virtual plane, in the sense of maximizing the throughput of VIPs in the network ${\cal G} = ({\cal N}, {\cal L})$ with appropriate transmission rate constraints.


\subsection{Maximizing VIP Throughput}

We now show that Algorithm~\ref{Alg:VIP} adaptively maximizes the throughput of VIPs in the network ${\cal G} = ({\cal N}, {\cal L})$ with appropriate transmission rate constraints.
In the following, we assume that (i) the VIP arrival processes $\{A^k_n(t); t=1,2,\ldots\}$  are mutually independent with respect to $n$ and $k$; (ii)
for all $n \in {\cal N}$ and $k \in {\cal K}$, $\{A^k_n(t); t = 1, 2, \ldots\}$ are i.i.d. with respect to $t$; and (iii) for all $n$ and $k$, $A^k_n(t)\leq A^k_{n,\max}$ for all $t$.

To determine the constraints on the VIP transmission rates $\mu^{k}_{ab}(t)$, we note that
Data Packets for the requested data object must travel on the reverse path taken by the Interest Packets.  Thus, in determining the transmission of the VIPs, we take into account the link capacities on the reverse path as follows:
\begin{align}
&\sum_{k \in \mathcal K} {\mu^{k}_{ab}(t)} \leq C_{ba}/z,\
\text{for~all}~(a,b) \in \mathcal L\label{eqn:rout_cost_sum}\\
&\mu^{k}_{ab}(t)=0, \;\text{for~all}~(a,b)\not\in \mathcal
L^{k}\label{eqn:rout_cost_non_neg}
\end{align}
where $C_{ba}$ is the capacity of ``reverse" link $(b,a)$.

\subsubsection{VIP Stability Region}

To present the throughput optimality argument,
we first define the VIP stability region.
The VIP queue at node $n$ is  {\em stable} if
$$ \limsup_{t \rightarrow \infty}
\frac{1}{t} \sum_{\tau = 1}^{t} 1_{[V^k_{n}(\tau) > \xi]} d\tau
\rightarrow 0 \;\; \text{as} \;\; \xi \rightarrow \infty,$$ where $1_{\{\cdot\}}$
is the indicator function.
The {\em VIP network stability region} $\Lambda$ is the closure
of the set of all VIP arrival rates $(\lambda^k_n)_{k \in {\cal K}, n \in {\cal N}}$ for which there exists some feasible joint forwarding and caching policy which can guarantee that all VIP queues are stable.  By feasible, we mean that at each time $t$, the policy specifies a forwarding rate vector $(\mu^k_{ab}(t))_{k \in {\cal K}, (a,b) \in {\cal L}}$ satisfying~\eqref{eqn:rout_cost_sum}-\eqref{eqn:rout_cost_non_neg}, and a caching vector $(s^k_n(t))_{k \in {\cal K}, n \in {\cal N}}$ satisfying the cache size limits $(L_n)_{n \in {\cal N}}$.

{The following theorem characterizes the VIP stability region in the virtual plane (or equivalently the IP stability region in the actual plane when there is no collapsing or suppression at the PITs), under the assumption that at each time $t$, each node $n \in {\cal N}$ can gain access to any data object $k \in {\cal K}$ and cache the object locally.} To our knowledge, Theorem~\ref{thm:stability} is the first instance where the effect of caching has been fully incorporated into the stability region of a multi-hop network.

\begin{Thm} [VIP Stability Region] The VIP stability region of the network ${\cal G} = ({\cal N}, {\cal L})$ with link capacity constraints~\eqref{eqn:rout_cost_sum}-\eqref{eqn:rout_cost_non_neg}, and with VIP queue evolution~\eqref{eqn:queue_dyn}, is the set $\Lambda$ consisting of all VIP arrival rates $(\lambda_n^k)_{k \in {\cal K}, n \in {\cal N}}$  such that there exist flow variables $(f_{ab}^k)_{k \in {\cal K}, (a,b) \in {\cal L}}$ and storage variables $(\beta_{n,i,l})_{ n \in {\cal N}; i=1,\cdots, {K \choose l};\  l=0,\cdots, i_n \triangleq \lfloor L_n/z\rfloor}$
satisfying
\begin{align}
&f_{ab}^k\geq 0, \ f_{nn}^k=0, \  f_{src(k)n}^k=0, \quad  \forall a,b, n\in \mathcal N, \ k\in \mathcal K\label{eqn:stability-region-f1}\\
&f_{ab}^k=0, \quad \forall a,b\in \mathcal N, \ k\in \mathcal K,\  (a,b)\not \in \mathcal L^k\label{eqn:stability-region-f2}\\
&0\leq \beta_{n,i,l}\leq 1, \; i=1,\cdots, {K \choose l},\ l=0,\cdots, i_n, \ n\in \mathcal N\label{eqn:stability-region-beta}\\
&\lambda_n^k\leq \sum_{b\in \mathcal N}f^k_{nb}-\sum_{a\in \mathcal N}f^k_{an}+r_n\sum_{l=0}^{i_n}\sum_{i=1}^{{K \choose l}}\beta_{n,i,l}\mathbf 1[k\in \mathcal B_{n,i,l}], \nonumber \\
&\hspace{35mm} \forall n\in \mathcal N,\ k\in \mathcal K, n\neq src(k) \label{eq:sink} \\
&\sum_{k\in \mathcal K} f^k_{ab}\leq C_{ba}/z, \quad\forall (a,b)\in \mathcal L\label{eqn:stability-region-capacity}\\
&\sum_{l=0}^{i_n}\sum_{i=1}^{{K \choose l}} \beta_{n,i,l}=1, \quad\forall n\in \mathcal N\label{eqn:stability-region-cache}
\end{align}
Here,
$\mathcal B_{n,i,l}$ denotes the caching set consisting of the $i$-th combination of $l$ data objects out of $K$ data objects at node $n$, where $i=1,\cdots, {K \choose l}$, $l=0,\cdots, i_n \triangleq\lfloor L_n/z\rfloor$.
\label{thm:stability}
\end{Thm}
\begin{proof} Please refer to Appendix A. 
\end{proof}

{To interpret Theorem~\ref{thm:stability}, note that the flow variable $f_{ab}^k$ represents the long-term VIP flow rate for data object $k$ over link $(a,b)$.  The
storage variable $\beta_{n,i,l}$ represents the long-term fraction of time that the set $\mathcal B_{n,i,l}$ (the $i$-th combination of $l$ data objects out of $K$ data objects) is cached at node $n$.  Inequality~\eqref{eq:sink} states that the (exogenous) VIP arrival rate for data object $k$ at node $n$ is upper bounded by the total long-term outgoing VIP flow rate minus the total (endogenous) long-term incoming VIP flow rate, plus the long-term VIP flow rate which is absorbed by all possible caching sets containing data object $k$ at node $n$, weighted by the fraction of time each caching set is used.}

\subsubsection{Throughput Optimality}

By definition, if the VIP arrival rates ${\boldsymbol \lambda} =
(\lambda^k_{n})_{k\in {\cal K}, n \in {\cal N}}$ $\in {\rm
int}({\Lambda})$, then all VIP queues can be stabilized. In general,
however, this may require knowing the value of ${\boldsymbol \lambda}$.  In
reality, ${\boldsymbol \lambda}$ can be learned only over time, and may be
time-varying.  Moreover, stabilizing the network given an arbitrary VIP arrival rate in the interior of $\Lambda$
may require (time sharing among) multiple forwarding and caching policies.

We now show that the joint forwarding and caching policy in Algorithm~\ref{Alg:VIP} adaptively stabilizes all VIP queues in the network for any ${\boldsymbol \lambda} \in {\rm
int}({\Lambda})$, without knowing ${\boldsymbol \lambda}$.  Thus, the policy is {\em throughput optimal}, in the sense of adaptively maximizing the VIP throughput, and therefore the user demand rate satisfied by the network.

\begin{Thm} [Throughput Optimality] If there exists  $\boldsymbol \epsilon=(\epsilon_n^k)_{n \in {\cal N}, k \in {\cal K}} \succ \mathbf 0$
such
that $\boldsymbol \lambda+\boldsymbol \epsilon \in \Lambda$,
then the network of VIP queues under
Algorithm \ref{Alg:VIP} satisfies
\begin{align}
\limsup_{t\to\infty}\frac{1}{t}\sum_{\tau=1}^{t}\sum_{n\in \mathcal N,k\in \mathcal K} \mathbb
E[V^k_n(\tau)]\leq \frac{N B}{\epsilon}\label{eqn:enhanced-DBP-VB}
\end{align}
where
$ B\triangleq \frac{1}{2N}\sum_{n\in \mathcal N}\big((\mu^{out}_{n,
\max})^2+(A_{n,\max}+\mu^{in}_{n,\max}+r_{n,\max})^2
+2\mu^{out}_{n,
\max}r_{n,\max}\big)$,
$\epsilon\triangleq \min_{ n\in \mathcal
N,k\in \mathcal K} \epsilon_n^k$,
$\mu^{in}_{n, \max}\triangleq \sum_{a\in \mathcal N} C_{an}/z $, $\mu^{out}_{n, \max}\triangleq
\sum_{b\in \mathcal N} C_{nb}/z$, $
A_{n,\max} \triangleq \sum_{k\in \mathcal K}
A^k_{n,\max}$, and $r_{n,\max}= K r_n$.\label{Thm:thpt-opt}
\end{Thm}
\begin{proof} Please refer to~\cite{viparxiv}.
\end{proof}

The forwarding and caching policy in Algorithm~\ref{Alg:VIP} achieves throughput optimality in the virtual plane by exploiting both the bandwidth and storage resources of the network to maximally balance out the VIP load (or the demand for data objects in the network), thereby preventing the buildup of congestion.
Equivalently,  Algorithm~\ref{Alg:VIP} is throughput optimal in the actual plane when Interest Packets are not collapsed or suppressed.
Note that Theorem~\ref{Thm:thpt-opt} can be seen as the multi-hop generalization of the throughput optimal result in~\cite{CachingTM2011:Ying}.
%

%

%

\subsection{Forwarding and Caching in the Actual Plane}
\label{sec:actual}


We now focus on the development of forwarding and caching policies for the actual plane, based on the throughput optimal policies of Algorithm~\ref{Alg:VIP} for the virtual plane.
Forwarding and caching in the actual plane take advantage of the exploration in the virtual plane to forward Interest Packets on profitable routes and cache Data Packets at profitable node locations.

\subsubsection{Forwarding of Interest Packets}

\label{sec:forwarding_actual}

The forwarding of Interest Packets in the actual plane follows the pattern established by the VIPs under Algorithm~\ref{Alg:VIP} in the virtual plane.
For a given window size $T$, let
\begin{align}   \bar{\nu}^k_{ab}(t) = \frac{1}{T} \sum_{t' = t-T+1}^t \nu^k_{ab}(t') \label{eq:window}
\end{align}
be the average number of VIPs for object $k$ {\em transmitted} over link $(a,b)$ over a sliding window of size $T$ under Algorithm~\ref{Alg:VIP} prior to time slot $t$.\footnote{Note that
the number $\nu^k_{ab}(t)$ of VIPs for object $k$ transmitted over link $(a,b)$ during time slot $t$ may not be the same as the allocated transmission rate $\mu^k_{ab}(t)$.  $\nu^k_{ab}(t)$ may be less than $\mu^k_{ab}(t)$ if there are few VIPs waiting to be transmitted.}


\textbf{Forwarding}:
At any node $n \in {\cal N}$, Interest Packets for all data objects share one queue and are served on a First-Come-First-Serve basis.
Suppose that the head-of-the-queue Interest Packet at node $n$ at time $t$ is an interest for the {\em starting chunk} of data object $k$.  If (i)
node $n$ has not yet received a request for data object $k$, or if the last type-$k$ data chunk in the last Data Packet received at node $n$ prior to $t$ is the {\em ending chunk} of object $k$, {\em and} if (ii)  there is no PIT entry at node $n$ for any chunk of data object $k$, then
forward the Interest Packet to node
\begin{align}
b_n^k(t) \in \arg \max_{\{b:(n,b)\in \mathcal L^k\}}\bar{\nu}^k_{nb}(t).\label{eqn:forwarding-RIP}
\end{align}
That is, the Interest Packet is forwarded on the link with the maximum average object-$k$ VIP flow rate over a sliding window of size $T$ prior to $t$, under Algorithm~\ref{Alg:VIP}.  This latter link is a ``profitable" link for forwarding the Interest Packet at time slot $t$, from the standpoint of reducing delays and congestion.
If either condition (i) or (ii) does not hold, then forward the Interest Packet on the link used by node $n$ to forward the most recent Interest Packet for a chunk of object $k$.\footnote{The router nodes need not know the names of the starting chunk and ending chunk beforehand.  These names can be learned as the routers forward Interest Packets and receive Data Packets for the popular data objects.  Before the names of the starting and ending chunks are learned, Interest Packets for the data object can be forwarded using a simple technique such as the shortest path algorithm.}

If the head-of-the-queue Interest Packet at node $n$ at time $t$ is an interest for a chunk of data object $k$ which is not the starting chunk, then forward the Interest Packet on the link used by node $n$ to forward the most recent Interest Packet for a chunk of object $k$.

The above forwarding algorithm ensures that a new request for data object $k$ (which does not overlap with any ongoing request for object $k$) at time $t$ is forwarded on the link with the maximum average object-$k$ VIP flow rate over a sliding window of size $T$ prior to $t$.  At the same time, the algorithm ensures that an ongoing request for data object $k$ keeps the same outgoing link from node $n$.  This ensures that in the actual plane, all the Interest Packets for an ongoing request for data object $k$ are forwarded on the same path toward a caching point or content source for data object $k$.
As a direct result, the Data Packets for all chunks for the same ongoing request for data object $k$ take the same reverse path through the network.

Note that the Interest Packets for non-overlapping requests for data object $k$ can still be forwarded on different paths, since the quantity $b_n^k(t)$ can vary with~$t$.  Thus, the forwarding of data object requests is inherently {\em multi-path} in nature.

It can be seen that the computational complexity (per time slot) of both the averaging operation in~\eqref{eq:window} and the link selection operation in~\eqref{eqn:forwarding-RIP} is $O(N^2K)$.  Thus, the complexity of forwarding (per time slot) in the actual plane is $O(N^2K)$.



\subsubsection{Caching of Data Packets}

As mentioned in Section~\ref{sec:vipframe}, within the instantiations of the VIP framework we consider, the caching algorithm in the actual plane coincides with the caching algorithm in the virtual plane.  Thus, in the current context, the caching algorithm for the actual plane is the same as that described in~\eqref{eqn:greedy-VIP}.  Thus, at each time slot $t$, the data objects with the highest VIP counts (highest local popularity) are cached locally.\footnote{For practical implementation in the actual plane, we cannot assume that at each time, each node can gain access to the data object with the highest local popularity for caching.  Instead, one can use a scheme similar to that discussed in Section~\ref{sec:algorithm2caching}, based on comparing the VIP count of the data object corresponding to a Data Packet received at a given node to the VIP counts of the data objects currently cached at the node.}

In attempting to implement the caching algorithm in~\eqref{eqn:greedy-VIP}, however, we encounter a problem. Since the VIP count of a data object is decremented by $r_n$ immediately after the caching of the object at node $n$, the strategy in~\eqref{eqn:greedy-VIP} exhibits oscillatory caching behavior, whereby data objects which are cached are shortly after removed from the cache again due to the VIP counts of other data objects now being larger.  Thus, even though Algorithm~\ref{Alg:VIP} is throughput optimal in the virtual plane, its mapping to the actual plane leads to policies which are difficult to implement in practice.

In the next section, we demonstrate another instantiation of the VIP framework yielding a forwarding and caching policy for the actual plane, which has more stable caching behavior.

\section{Stable Caching VIP Algorithm}
\label{sec:forwarding-caching-stable}


In this section, we describe a practical VIP algorithm, called Algorithm 2, that looks for a stable solution in which the cache contents do not cycle in steady-state. Although Algorithm 2 is not theoretically optimal in the virtual plane, we show that it leads to significant performance gains in simulation experiments.

\subsection{Forwarding of Interest Packets}\label{sec:forwarding-caching-RIP}

The forwarding algorithm in the virtual plane for Algorithm 2 coincides with the backpressure-based forwarding scheme described in
\eqref{eqn:forwarding-VIP}-\eqref{eqn:weight} for Algorithm 1.   The forwarding of Interest Packets in the actual plane for Algorithm 2 coincides with the forwarding scheme described in~\eqref{eqn:forwarding-RIP}.
That is, all the Interest Packets for a particular request for a given data object are forwarded on the link with the maximum average VIP flow rate over a sliding window of size $T$ prior to the arrival time of the Interest Packet for the first chunk of the data object.

\subsection{Caching of Data Packets}

\label{sec:algorithm2caching}

The caching decisions are based
on the VIP flow in the virtual plane.
Suppose that at time slot $t$, node $n$ receives the Data Packet containing the first chunk of data object $k_{new}$ which is not currently cached at node $n$.   If there is sufficient unused space in the cache of node $n$ to accommodate the Data Packets of all chunks of object $k_{new}$, then node $n$ proceeds to cache the Data Packet containing the first chunk of data object $k_{new}$
as well as the Data Packets containing all subsequent chunks for data object $k_{new}$ (which, by the forwarding algorithm in Section~\ref{sec:forwarding_actual}, all take the same reverse path through node $n$).  That is, the entire data object $k$ is cached at node $n$.
Otherwise, the node compares the {\em cache scores} for $k_{new}$ and the currently cached objects, as follows.  For a given window size $T$, let the cache score for object $k$ at node $n$ at time $t$ be
\begin{equation}
  CS^k_n(t) = \frac{1}{T} \sum_{t' = t-T+1}^t \sum_{(a,n) \in {\cal L}^k}\nu^k_{an}(t')=\sum_{(a,n)\in {\cal L}^k}\bar{\nu}^k_{an}(t),  \label{eq:cache_score}
\end{equation}
i.e., the average number of VIPs for object $k$ {\em received} by node $n$ over a sliding window of size $T$ prior to time slot $t$.
Let ${\cal K}_{n,old}$ be the set of objects that are currently cached at node $n$.  Assuming that all data objects are of equal size, let $k_{min} \in {\cal K}_{n,old}$ be a current cached object with the smallest cache score.  If $k_{new}$ has a lower cache score than $k_{min}$, then object $k_{min}$ (consisting of all chunks) is evicted and replaced with object $k_{new}$. Otherwise, the cache is unchanged. If objects have different sizes, the optimal set of objects is chosen to maximize the total cache score under the cache space constraint. This is a knapsack problem for which low complexity heuristics exist.

At each time $t$, the VIP count at node $n$ for object $k$ is decreased by $r_n s_n^{k}(t)$ due to the caching at node $n$.
{This has the effect of attracting the flow of VIPs for each object $k\in\mathcal K_{n,new}$, where ${\cal K}_{n,new}$ denotes the new set of cached objects, to node $n$.}  

{The Data Packets for data objects evicted from the cache are potentially cached more efficiently elsewhere} (where the demand for the evicted data object is relatively bigger).   This is realized as follows: before the data object is evicted, VIPs and Interest Packets flow toward the caching point as  it is a sink for the  object.  After eviction, the VIP count would begin building up since the VIPs would not exit at the caching point.  As the VIPs build further, the backpressure load-balancing forwarding policy would divert them away from the current caching point to other parts of the network.

We now find the caching complexity for Algorithm 2.  Note that the complexity of calculating the cache scores (per time slot) in~\eqref{eq:cache_score} is $O(N^2K)$.  Due to link capacity constraints, the number of new data objects which arrive at a given node in a time slot is upper bounded by a constant.  Thus, for fixed cache sizes, the total computational complexity for the cache replacement operation (per time slot) is $O(N)$.  In sum, the caching complexity for Algorithm 2 per time slot is $O(N^2K)$.

\begin{figure*}[t]
\begin{center}
\includegraphics[width=17.4cm, height=8.4cm]{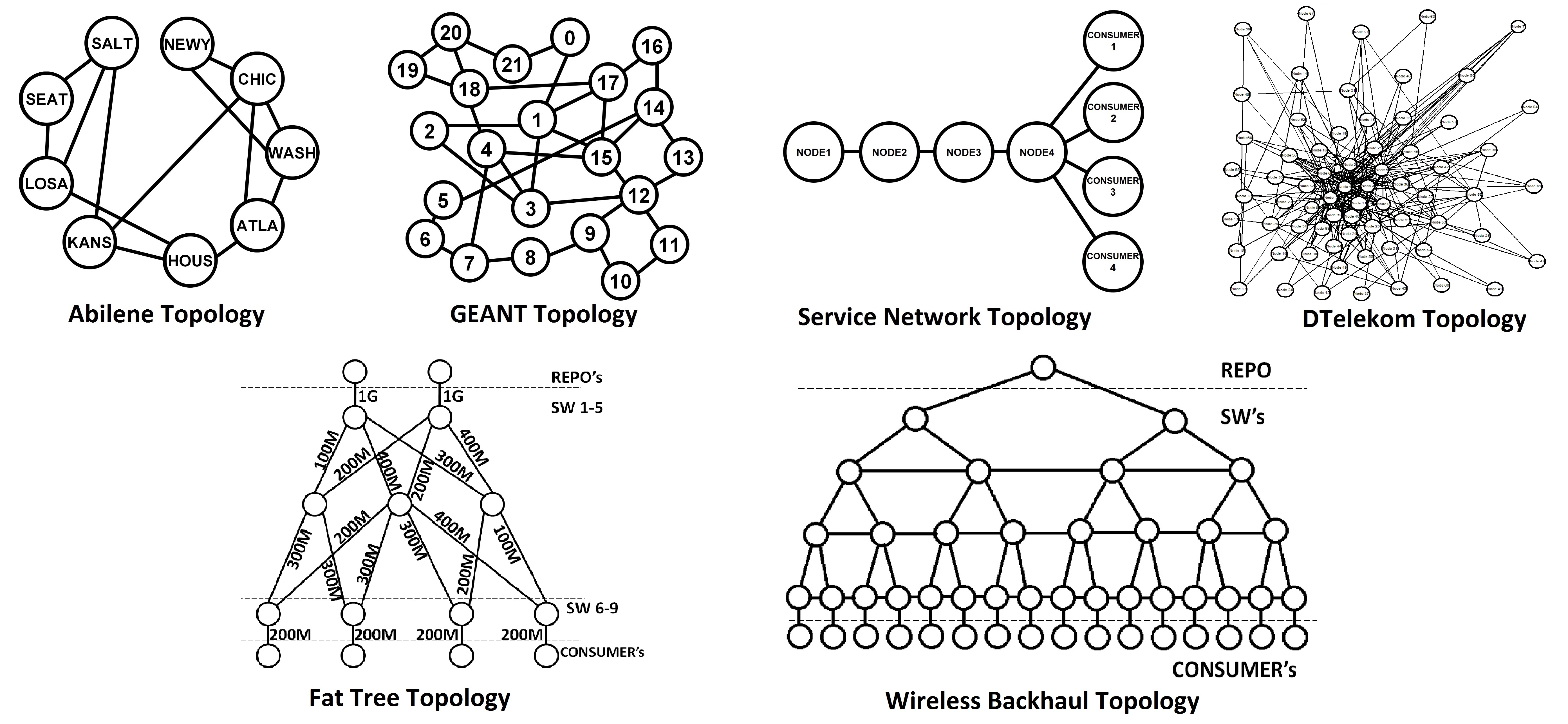}
\caption{Network Topologies}
\end{center}
\label{fig:topology}
\end{figure*}

\section{VIP Congestion Control}\label{sec:congestion-VIP}

Even with optimal forwarding and caching, excessively large request rates can overwhelm a content-centric network with limited resources.  Developing effective, practical, and fair congestion control algorithms for content-centric networks is a new and important challenge.  The VIP framework can effectively address this challenge by optimally combining congestion control with forwarding and caching.
Note that there is a downward VIP ÒgradientÓ from entry points of the object requests to the content source and caching nodes.  Thus, a high VIP count at the request entry point provides a natural signal for congestion control.
Based on the approach in~\cite{Georgiadis-Neely-Tassiulas:2006}, we develop a scheme where newly arriving Interest Packets (equivalently VIPs) first enter transport-layer storage reservoirs before being admitted to network-layer queues by a congestion control algorithm.  For each object $k$ and entry node $n$, associate a non-decreasing and concave utility function $g^k_n(\cdot)$ with VIPs admitted into the network layer.\footnote{Note that various fairness notions can be reflected via appropriate choice of the utility functions.}  When exogenous VIP arrival rates are outside the stability region $\Lambda$, the goal is to support a portion of the VIPs which maximizes the sum of utilities (over all $k$ and $n$) while ensuring network-layer stability.  We  present a distributed, joint congestion control, forwarding, and caching algorithm which adaptively and dynamically uses network bandwidth and storage resources to maximize the aggregate utility of admitted request rates subject to network-layer stability, for all exogenous VIP arrival rates inside or outside the stability region, even without knowledge of the arrival rates.  The algorithm also yields a characteristic tradeoff between the attained aggregate utility and the overall network delay.


\subsection{Transport Layer and Network Layer VIP Dynamics}

Let $Q_{n,\max}^k$ and $Q_n^k(t)$ denote  the transport layer VIP buffer size and the VIP count for object $k$ at node $n$ at the beginning of slot $t$, respectively.  $Q_{n,\max}^k$ can be infinite or finite (possibly zero\footnote{When $Q_{n,\max}^k=0$,  $Q_n^k(t)=0$ for all $t$.}). Let $\alpha^k_n(t)\geq 0$ denote the amount of VIPs admitted to the network layer VIP queue of object $k$ at node $n$ from the transport layer VIP queue at slot $t$. We assume $\alpha^k_n(t)\leq \alpha^k_{n,\max}$, where $\alpha^k_{n,\max}$ is a   positive constant which limits the burstiness of the admitted VIPs to the network layer.  We have  the following time evolutions of the transport  and network layer  VIP counts:
\begin{align}
Q^k_n(t+1) = \min \left\{ \left(Q^k_n(t)-\alpha^k_n(t)\right)^+ +A^k_n(t), Q^k_{n,\max} \right\}
\label{eqn:queue_dyn-trans}
\end{align}
\begin{align}
& V^k_n(t+1) \leq \nonumber\\
&   \left(
\left(V^k_n(t)-\sum_{b\in \mathcal N}\mu^{k}_{nb}(t)\right)^+ +\alpha^k_n(t)
+\sum_{a\in \mathcal N}\mu^{k}_{an}(t)- r_n s_n^{k}(t)\right)^+\label{eqn:queue_dyn-flow}
\end{align}

\subsection{Congestion Control Algorithm}

The goal of the congestion control is to support a portion of the VIPs  which maximizes the sum of utilities when $\boldsymbol \lambda  \notin \Lambda$.  Let $g^k_n(\cdot)$ be the utility function associated with the VIPs into the network layer for object $k$ at node $n$.  Assume  $g^k_n(\cdot)$ is non-decreasing, concave, continuously differentiable and non-negative. Define a $\boldsymbol \theta$-optimal admitted VIP rate as follows:
\begin{align}
\overline {\boldsymbol \alpha}^*(\boldsymbol \theta)=\arg\max_{\overline{\boldsymbol \alpha}}\quad  & \sum_{n\in \mathcal N,k\in \mathcal K}g^k_n\left(\overline \alpha^k_n\right)\label{eqn:eps-opt-prob}\\
s.t. \quad &  \overline{\boldsymbol \alpha}+\boldsymbol \theta \in \Lambda\label{eqn:stable}\\
& \mathbf 0 \preceq \overline{\boldsymbol \alpha} \preceq \boldsymbol \lambda\label{eqn:demand}
\end{align}
where $ \overline{\boldsymbol \alpha}^*(\boldsymbol \theta)=(\overline \alpha^{k*}_n(\boldsymbol \theta))$, $ \overline{\boldsymbol \alpha}=(\overline \alpha^k_n)$ and $\mathbf 0\preceq\boldsymbol \theta=(\theta^k_n)\in \Lambda $.  The constraint in \eqref{eqn:stable} ensures that the admitted rate to the network layer is bounded away from the boundary of the network stability region by $\boldsymbol \theta$.
Due to the non-decreasing property of the utility functions, the maximum sum utility over all $\boldsymbol \theta$ is achieved at $\overline {\boldsymbol \alpha}^*(\mathbf 0)$ when  $\boldsymbol \theta=\mathbf 0$.

In the following, we develop a   joint congestion control, forwarding and caching algorithm that yields a throughput vector which can be arbitrarily close to the optimal solution $\overline {\boldsymbol \alpha}^*(\mathbf 0)$. We introduce auxiliary variables $\gamma_n^k(t)$ and the virtual queues $Y_n^k(t)$ for all $n\in \mathcal N$ and $k\in \mathcal K$.\footnote{Note that the congestion control part  of Algorithm \ref{Alg:flow-DBP} is the same as that in the traditional joint congestion control and DBP algorithm in \cite[page 90]{Georgiadis-Neely-Tassiulas:2006}. The difference lies in the forwarding  and caching part. We describe the congestion control part here for the purpose of completeness. }

\setcounter{Alg}{2}
\begin{Alg}  Initialize the virtual VIP count $Y^k_n(0)=0$ for all $n\in\mathcal N$ and $k\in \mathcal K$.
At the beginning of  each time slot $t$,   observe the network layer VIP counts $(V^k_n(t))_{k\in \mathcal K, n \in \mathcal N}$ and virtual VIP counts $(Y^k_n(t))_{k\in \mathcal K, n \in \mathcal N}$, and performs the following congestion control, forwarding
and caching in the virtual plane as follows.

\textbf{Congestion Control}: For each node $n$ and object $k$, choose the admitted VIP count at slot $t$, which also serves as the output rate of the corresponding virtual queue:
\begin{align}
\alpha^k_n(t)=
\begin{cases}
\min\left\{Q^k_n(t), \alpha^k_{n,\max}\right\},  & Y^k_n(t)>V^k_n(t)\\
0,& \text{otherwise}
\end{cases}\nonumber
\end{align}
Then, choose the auxiliary variable, which serves as the input rate  to the corresponding virtual queue:
\begin{align}
\gamma^k_n(t)=\arg\max_{\gamma}\quad & W g^k_n(\gamma)-Y^k_n(t)\gamma\label{eqn:solu-gamma}\\
s.t. \quad & 0\leq \gamma\leq \alpha^k_{n,\max}\nonumber
\end{align}
where $W>0$ is a control parameter which affects the utility-delay tradeoff of the algorithm.
Based on the chosen $\alpha^k_n(t)$ and $\gamma^k_n(t)$, the  transport
layer VIP count is updated according to \eqref{eqn:queue_dyn-trans} and the virtual VIP count is  updated according to:
\begin{align}
Y^k_n(t+1)=& \left(Y^k_n(t)-\alpha^k_n(t)\right)^+ +\gamma^k_n(t)\label{eqn:queue_dyn-virtual}
\end{align}

\textbf{Forwarding and Caching}:  Same as Algorithm \ref{Alg:VIP}. The network layer VIP count is updated according to \eqref{eqn:queue_dyn-flow}.\label{Alg:flow-DBP}
\end{Alg}

\subsection{Utility-Delay Tradeoff}

We now show that  for any control parameter $W>0$,  the joint congestion control, forwarding and caching policy in Algorithm~\ref{Alg:flow-DBP} adaptively stabilizes all VIP queues in the network ${\cal G} = ({\cal N}, {\cal L})$ for any ${\boldsymbol \lambda} \not\in {\rm
int}({\Lambda})$, without knowing ${\boldsymbol \lambda}$.   Algorithm~\ref{Alg:flow-DBP} yields a throughput vector which can be arbitrarily close to the optimal solution $\overline {\boldsymbol \alpha}^*(\mathbf 0)$ by letting $W\to \infty$. Similarly, in the following, we assume that the VIP arrival processes  satisfy (i) for all $n \in {\cal N}$ and $k \in {\cal K}$, $\{A^k_n(t); t = 1, 2, \ldots\}$ are i.i.d. with respect to $t$; (ii) for all $n$ and $k$, $A^k_n(t)\leq A^k_{n,\max}$ for all $t$.

\begin{Thm}[Utility-Delay Tradeoff of Alg. \ref{Alg:flow-DBP}] For an arbitrary VIP arrival rate  $\boldsymbol \lambda$ and for any control parameter $W>0$,  the  network of VIP queues under
Algorithm \ref{Alg:flow-DBP}  satisfies
\begin{align}
\limsup_{t\to\infty}\frac{1}{t}\sum_{\tau=1}^{t}\sum_{n\in \mathcal N,k\in \mathcal K} \mathbb
E[V^k_n(\tau)]&\leq \frac{2N\hat B+WG_{\max}}{
2\hat \epsilon}\label{eqn:enhanced-flow-DBP-U}
\end{align}
\begin{align}
\liminf_{t\to\infty}\sum_{n\in \mathcal N,k\in \mathcal K} g^k_n\left(\overline \alpha^k_n(t)\right)\geq &
\sum_{n\in \mathcal N,k\in \mathcal K}
g^{(c)}_n\left( \overline \alpha^{k*}_n\left(\mathbf 0\right)\right)\nonumber\\
&-\frac{2N\hat B}{W}
\label{eqn:enhanced-flow-DBP-g}
\end{align}
\begin{align}
\hat B\triangleq& \frac{1}{2N}\sum_{n\in \mathcal N}\Big((\mu^{out}_{n,
\max})^2+(\alpha_{n,\max}+\mu^{in}_{n,\max}+r_{n,\max})^2\nonumber\\
&\hspace{15mm} +2\mu^{out}_{n,
\max}r_{n,\max}\Big)\label{eqn:bound-B-flow}\\
\hat \epsilon \triangleq &\sup_{\{\boldsymbol \epsilon:  \boldsymbol \epsilon\in  \Lambda\}}\min_{ n\in \mathcal N, k\in \mathcal K}\left\{\epsilon_n^k\right\}\label{eqn:bound-e-flow}
\end{align}
with $\alpha_{n,\max} \triangleq  \sum_{k\in \mathcal K}
\alpha^k_{n,\max}$,   $G_{\max}\triangleq\sum_{n\in \mathcal N,k\in \mathcal K}
g^k_n\left(\alpha^k_{n,\max}\right)$, $\overline \alpha^k_n(t)\triangleq \frac{1}{t}\sum_{\tau=1}^t\mathbb E[ \alpha^k_n(\tau)]$.
\label{Thm:flow-control}
\end{Thm}
\begin{proof} Please refer to Appendix C.
\end{proof}

\section{Experimental Evaluation}
\label{sec:simulations}

This section presents the experimental evaluation of the proposed VIP Algorithms.
Experimental scenarios are carried on six network topologies: the Abilene Topology (9 nodes), the GEANT Topology (22 nodes), the Service Network Topology (8 nodes), the  DTelekom Topology (68 nodes), the Fat Tree Topology (15 nodes) and the Wireless Backhaul Topology (47 nodes), as shown in Figure~2.

In the Abilene, GEANT, and  DTelekom topologies, object requests can be generated by any node, and the content source for each data object is independently and uniformly distributed among all nodes.
In the Service Network topology, NODE 1 is the content source for all objects and requests can be generated only by the CONSUMER nodes. Caches are placed at NODE 2, NODE 3, NODE 4 and the CONSUMER nodes.
In the Fat Tree Topology, the content source for each data object is independently and uniformly distributed between the REPO nodes.  Requests are generated by  the CONSUMER nodes.
In the Wireless Backhaul topology, REPO is the content source for all data objects.  Requests are generated only by CONSUMER nodes.
At each node requesting data, object requests arrive according to a Poisson process  with an overall rate $\lambda$ (in requests/node/sec). Each arriving request requests data object $k$ (independently) with probability $p_k$, where $\{p_k\}$ follows a (normalized) Zipf distribution with parameter 0.75.
In the simulations, the Interest Packet size is 125 $B$; the Data Packet size is 50 $KB$; the data object size is 5 $MB$.
We assume dropping of Interest and Data Packets due to buffer overflow does not occur.  We do not consider PIT expiration timers and interest retransmissions.

In the virtual plane, we use $\Delta$ to denote the time slot length. 
For the proposed VIP algorithms, the  forwarding in the virtual plane uses the backpressure algorithm with a cost bias\footnote{The cost bias is calculated as the number of hops on the shortest path to the content source, and is added to the VIP queue differential. It can be shown that the cost-biased version  is also throughput optimal in the virtual plane, as in Theorem~\ref{Thm:thpt-opt}.} to help direct VIPs toward content sources, and
the window size $T$ in the actual plane is 5000 time slots.
Each simulation generates requests  for  100 $sec$.  Each curve in Figures~3-19 is obtained by averaging over 10 simulation runs. Our simulations are carried out on the Discovery cluster from Northeastern University.

\subsection{Experimental Evaluation of Algorithm 2}

In this section, we present experimental results on the performance of the Stable Caching VIP Algorithm (Algorithm 2). Experimental scenarios are carried on all six network topologies.

The cache sizes at all nodes are identical, and are chosen to be 5 $GB$ (1000 data objects) in the Abilene topology, and 2 $GB$ (400 data objects) in the GEANT and DTelekom topologies. In the Service Network, the cache sizes at NODE 2, NODE 3, NODE 4 and the CONSUMER nodes are 5 $GB$.
In the Fat Tree Topology, the cache sizes are 1 $GB$ (200 data objects) at SW 1-5 and 125 $MB$ (25 data objects) at SW 6-9.
In the Wireless Backhaul topology, the cache sizes at all SW nodes are 100 $MB$ (20 data objects)
In the virtual plane, $\Delta=200$ msec in the GEANT and DTelekom topologies, $\Delta=80$ msec in the Service Network and Abilene topologies and $\Delta=40$ msec in the Fat Tree and Wireless Backhaul topologies. Each simulation terminates when all Interest Packets are fulfilled.


Simulation experiments were carried out to compare the Stable Caching VIP Algorithm against a number of popular caching algorithms used in conjunction with shortest path forwarding and a potential-based forwarding algorithm.  In shortest path forwarding, at any given node, an Interest Packet for data object $k$ is forwarded on the shortest path to the content source for object $k$.\footnote{We assume that all chunks of a data object are cached together.}
The Data Packet corresponding to the Interest Packet may be retrieved from a caching node along the shortest path.
In potential-based forwarding, a potential value for each data object at each node is set as in~\cite{PotentialRouting2012}.
At each time and for each node, an Interest Packet for object $k$ is forwarded to the neighbor with the lowest current potential value for object $k$.
%
Each caching algorithm consists of two parts: caching decision and caching replacement.  Caching decision determines whether or not to cache a new data object when the first chunk of this object arrives and there is no remaining cache space.  If a node decides to cache the new data object, then caching replacement decides which currently cached data object should be evicted to make room for the new data object.  We considered the following caching decision policies: (i) Leave Copies Everywhere (LCE), which decides to cache all new data objects, and (ii) Leave a Copy Down (LCD)\cite{LCD:1395054}, where upon a cache hit for data object $k$ at node $n$, object $k$ is cached at the node which is one hop closer to the requesting node (while object $k$ remains cached at node $n$).
We considered  the following caching replacement policies: (i) Least Recently Used (LRU), which replaces the least recently requested data object,
(ii) UNIF, which chooses a currently cached data object for replacement, uniformly at random, and (iii) BIAS, which chooses two currently cached data objects uniformly at random, and then replaces the less frequently requested one.  In addition, we considered Least Frequently Used (LFU) and age-based caching~\cite{Age-based:6193504}. In LFU, the nodes record how often each data object has been requested and choose to cache the new data object if it is more frequently requested than the least frequently requested cached data object (which is replaced).  A window size of 5000 time slots is used for LFU.
In age-based caching~\cite{Age-based:6193504}, each cached object $k$ at node $n$ is assigned an age which depends on $p_k$, the (Zipf) popularity of object $k$, and the shortest-path distance between $n$ and $src(k)$.  The cache replacement policy replaces the cached  object for which the age has been exhausted the longest.
\begin{figure}[t]
\begin{center}
\vspace{-\baselineskip}
\centering
\includegraphics[height=4.8cm]{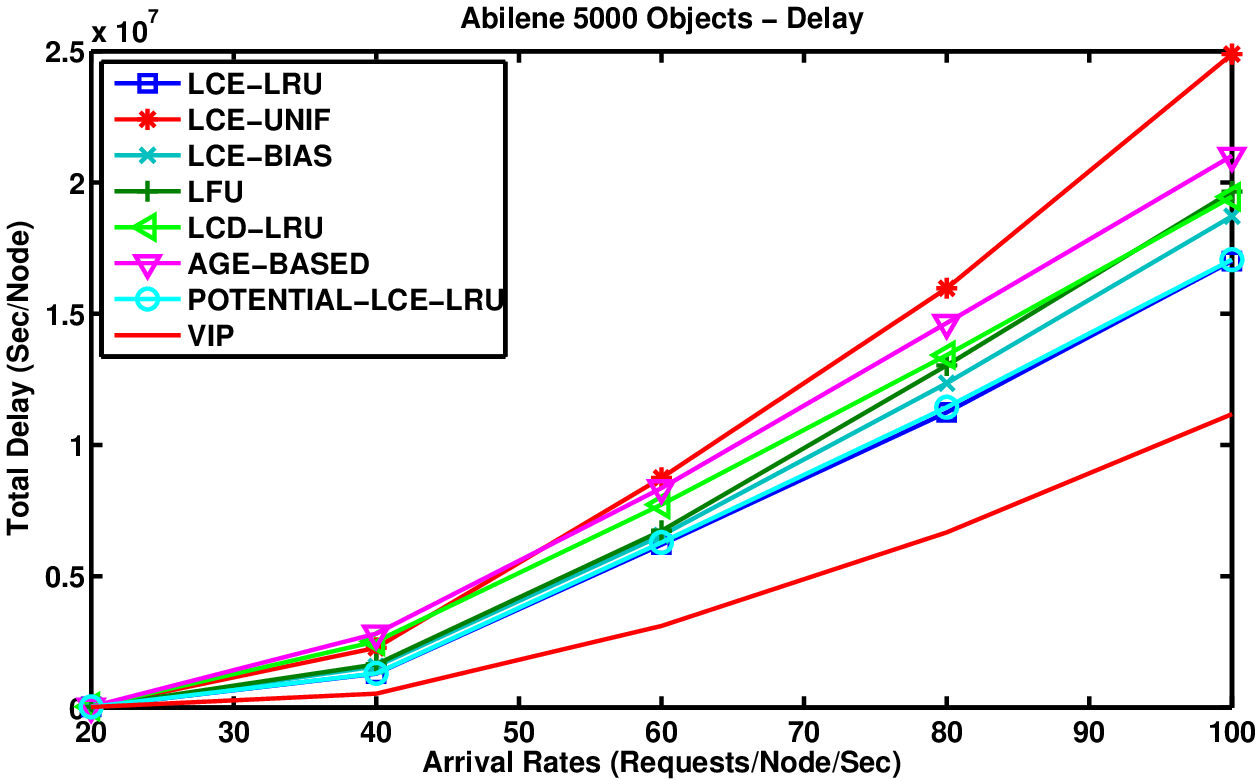}
\caption{Abilene Topology: Delay}

\vspace{10pt}
\centering
\includegraphics[height=4.8cm]{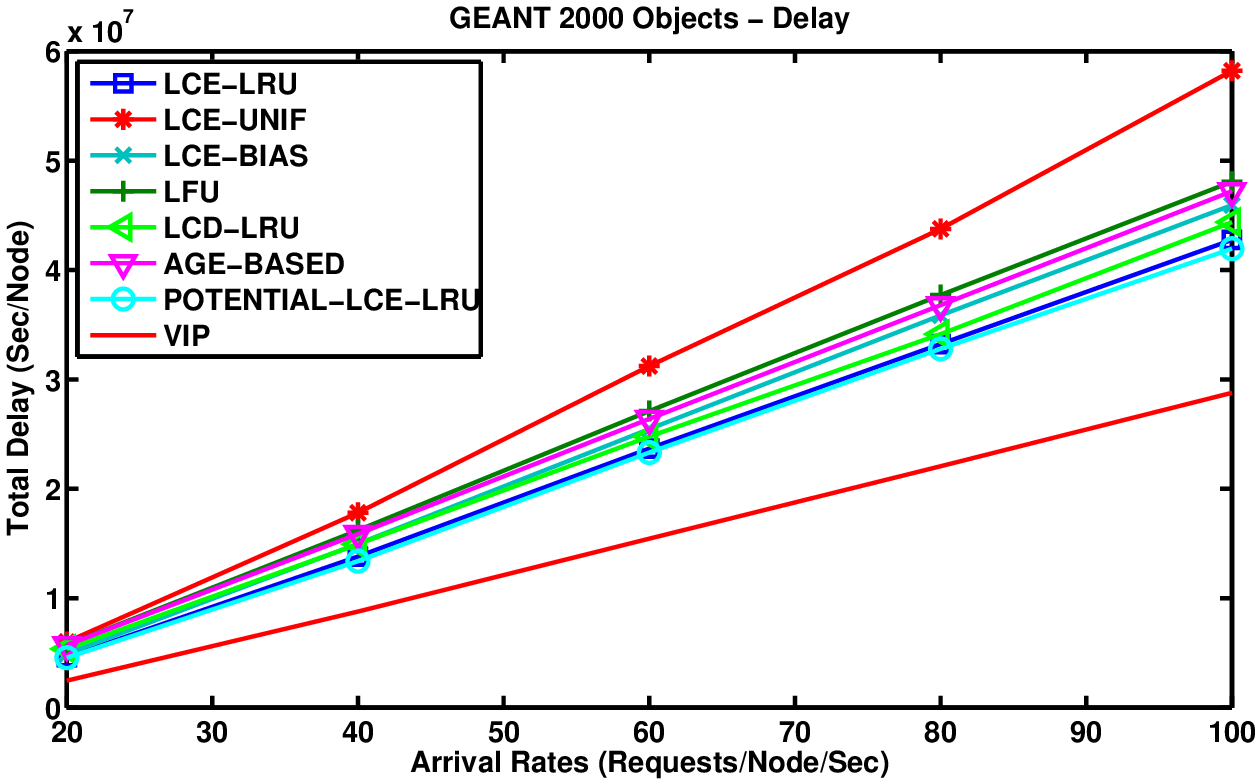}
\caption{GEANT Topology: Delay}

\vspace{10pt}
\centering
\includegraphics[height=4.8cm]{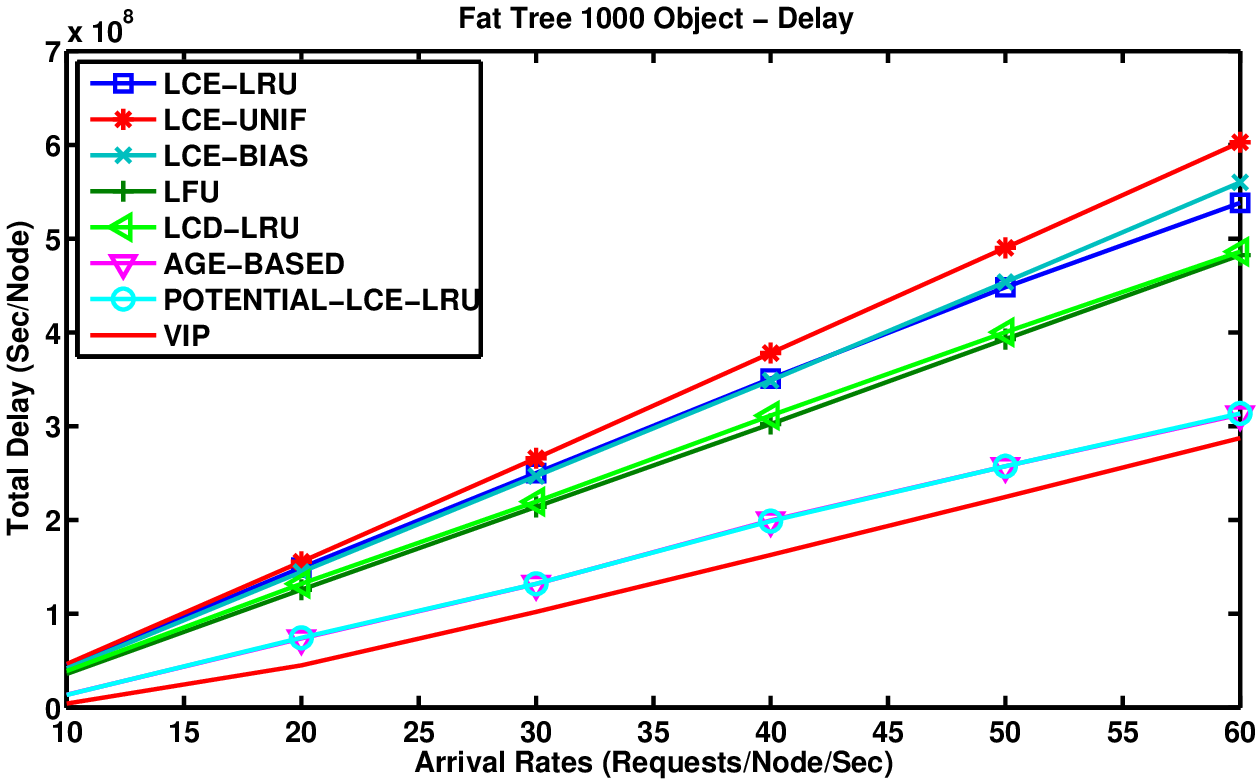}
\caption{Fat Tree Topology: Delay}
\end{center}
\end{figure}

\begin{figure}[t]
\begin{center}
\vspace{-\baselineskip}
\centering
\includegraphics[height=4.8cm]{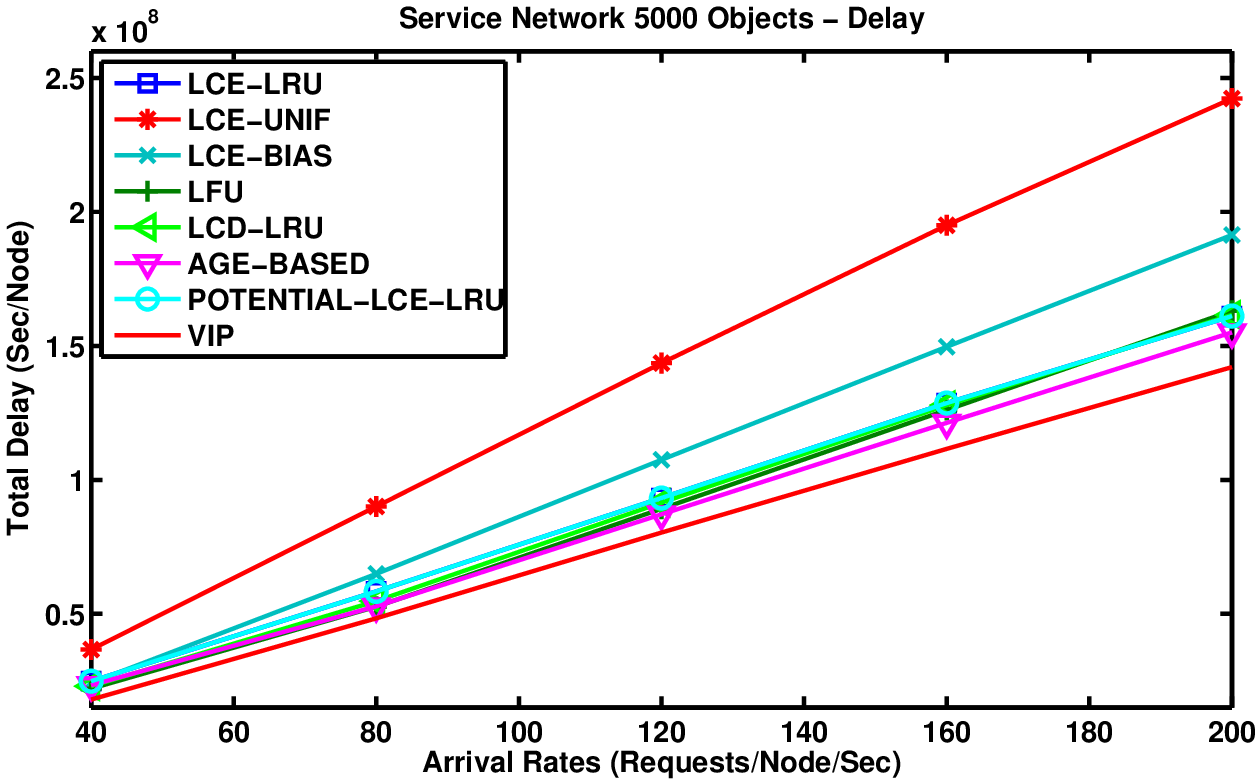}
\caption{Service Network Topology: Delay}

\vspace{10pt}
\centering
\includegraphics[height=4.8cm]{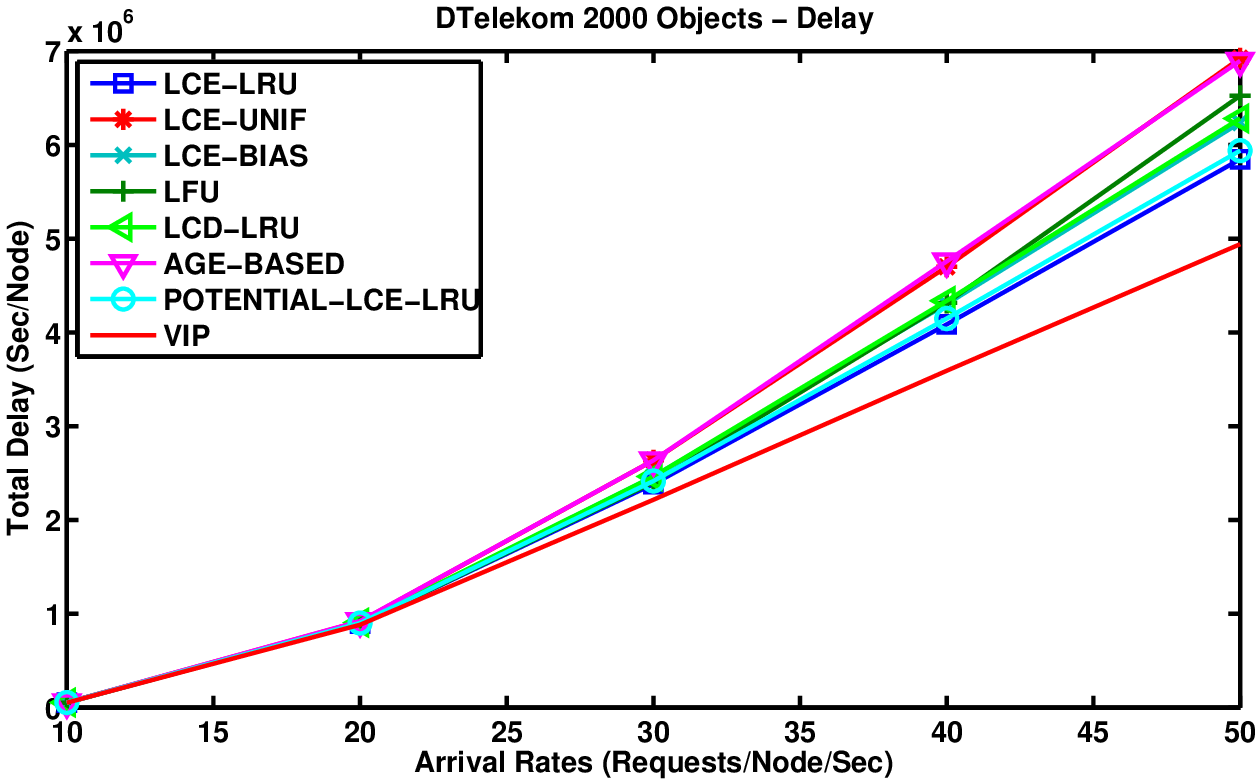}
\caption{DTelekom Topology: Delay}

\vspace{10pt}
\centering
\includegraphics[height=4.8cm]{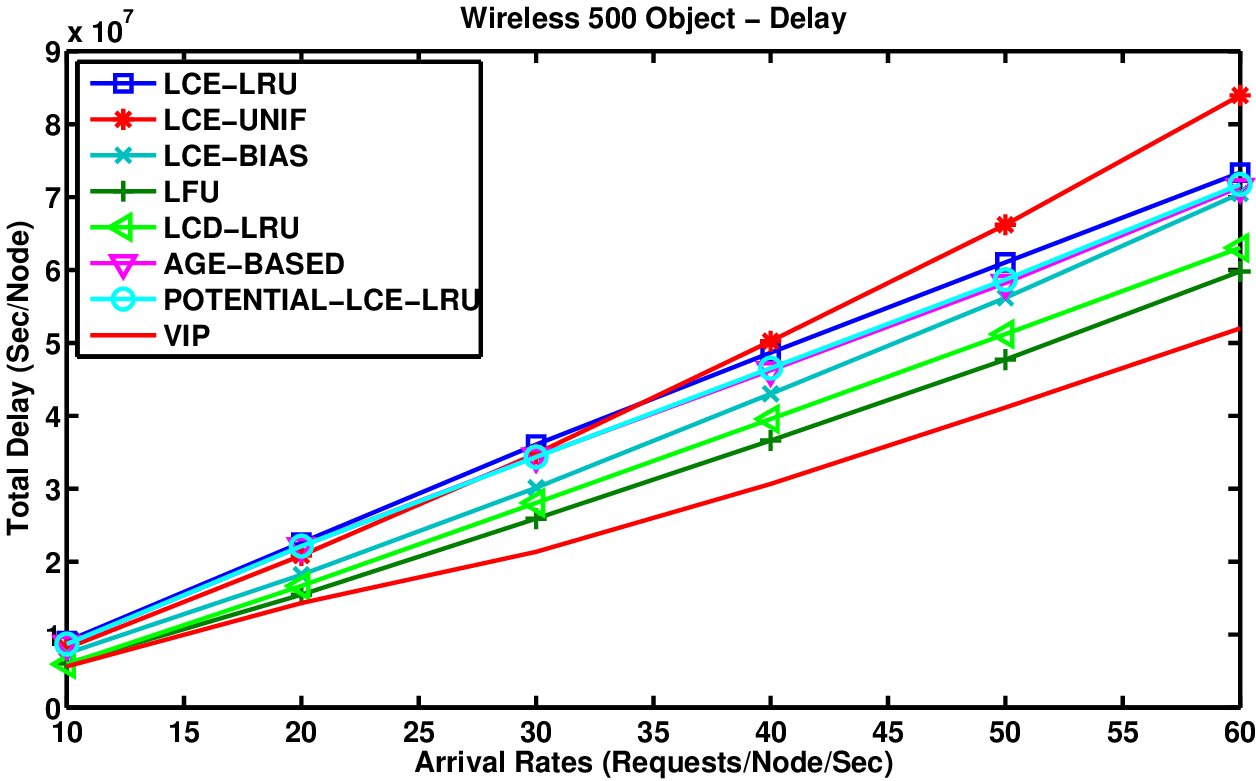}
\caption{Wireless Backhaul Topology: Delay}
\end{center}
\end{figure}

\begin{figure}[t]
\begin{center}
\vspace{-\baselineskip}
\centering
\includegraphics[height=4.8cm]{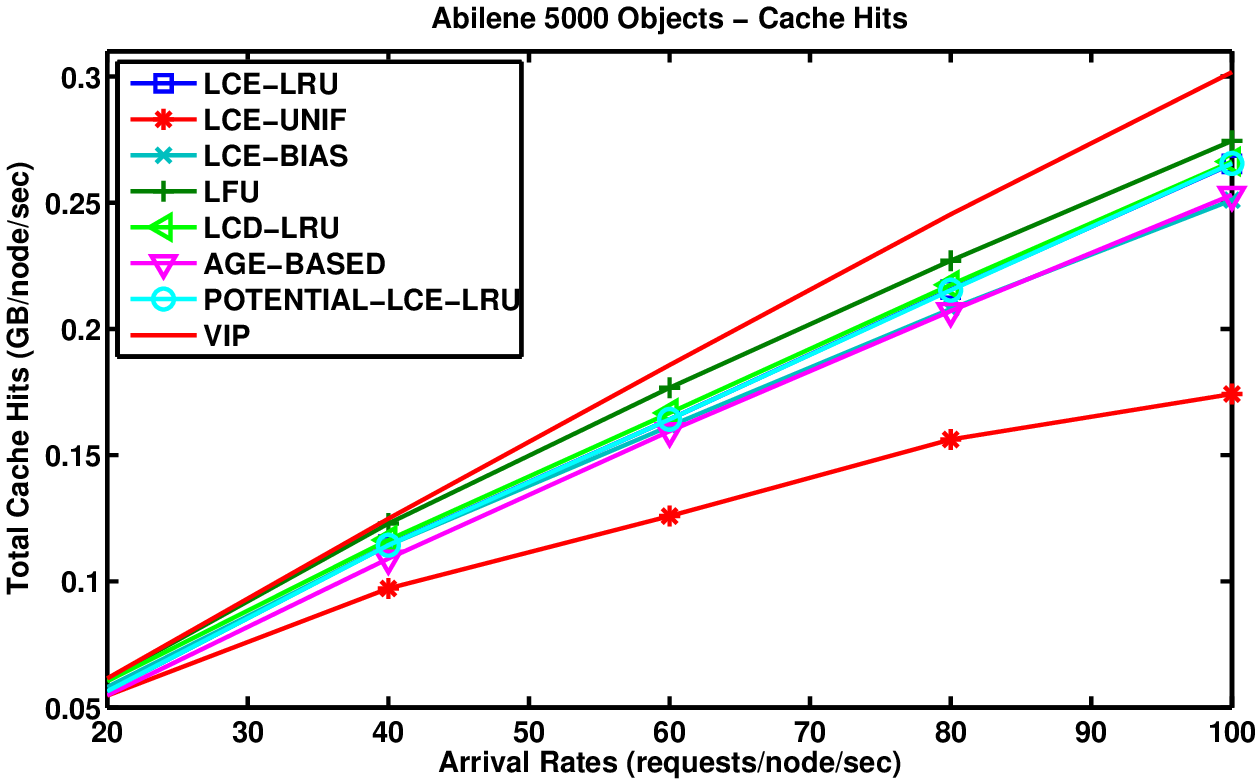}
\caption{Abilene Topology: Cache Hits}

\vspace{10pt}
\centering
\includegraphics[height=4.8cm]{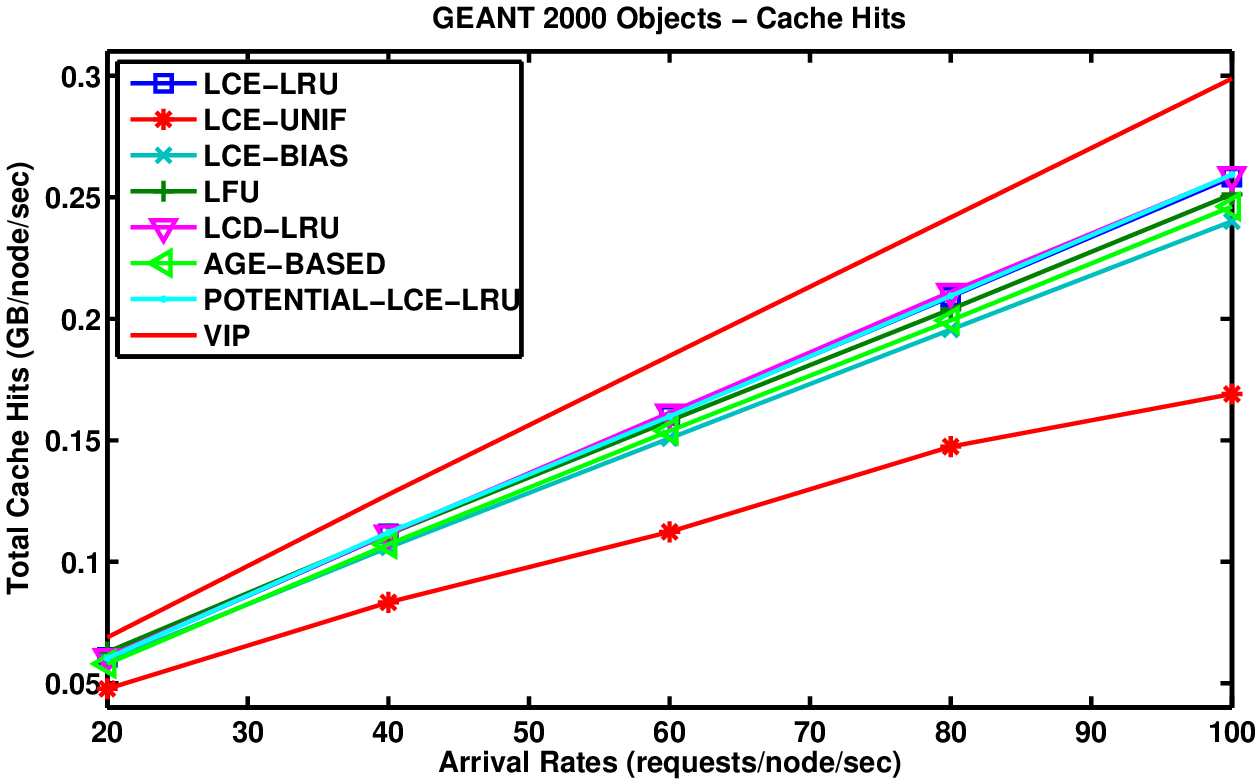}
\caption{GEANT Topology: Cache Hits}

\vspace{10pt}
\centering
\includegraphics[height=4.8cm]{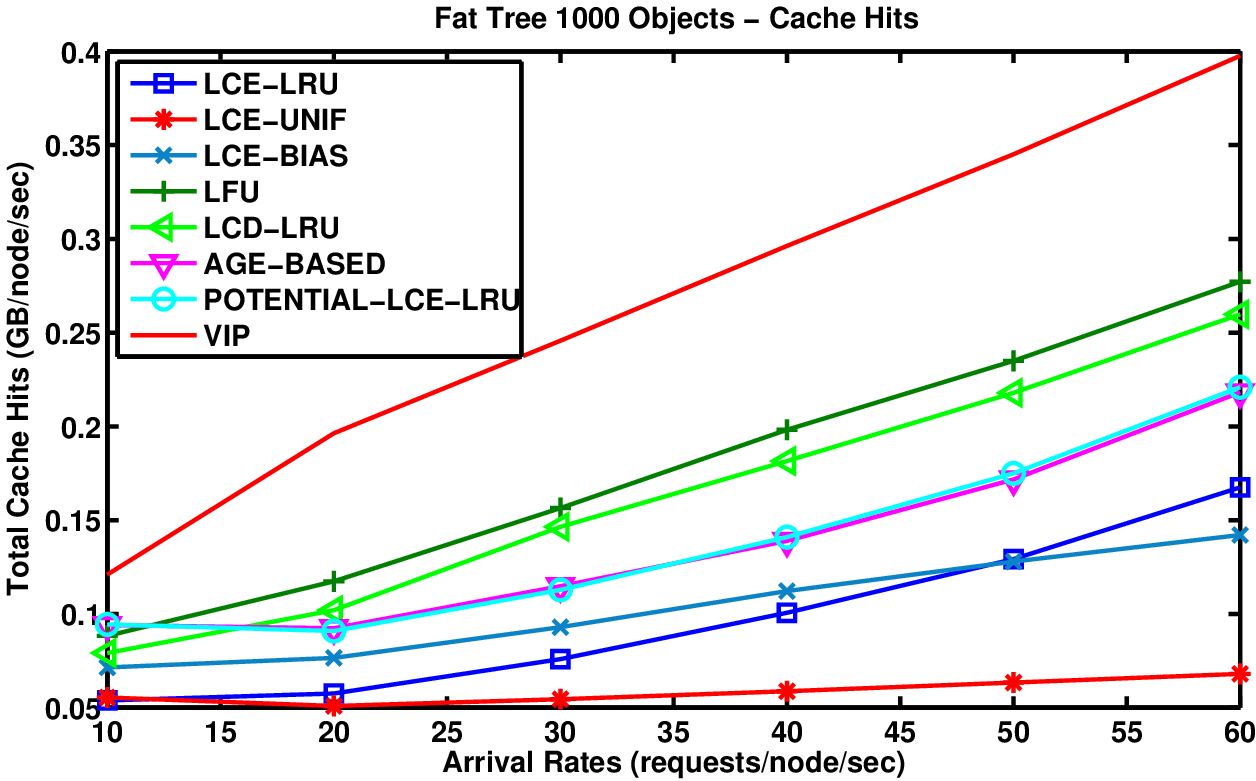}
\caption{Fat Tree Topology: Cache Hits}
\end{center}

\end{figure}

\begin{figure}[t]
\begin{center}
\vspace{-\baselineskip}
\centering
\includegraphics[height=4.8cm]{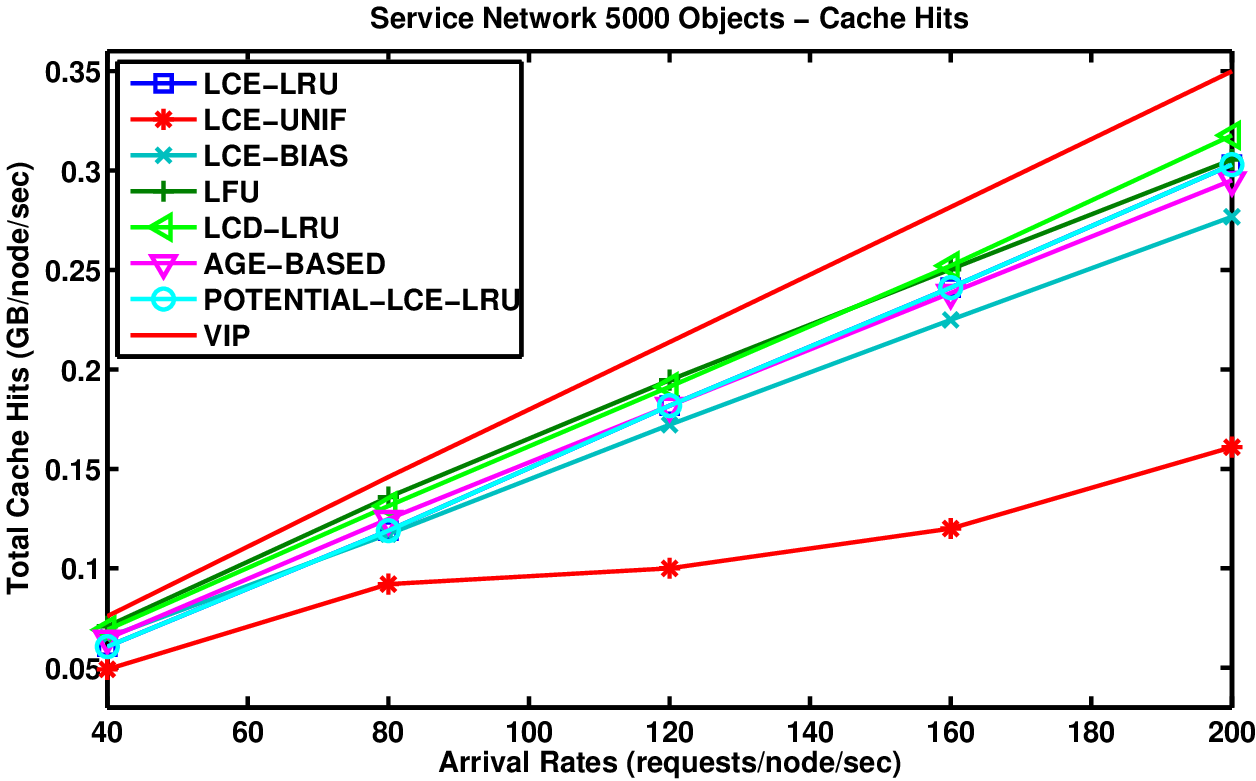}
\caption{Service Network Topology: Cache Hits}
\vspace{10pt}
\centering
\includegraphics[height=4.8cm]{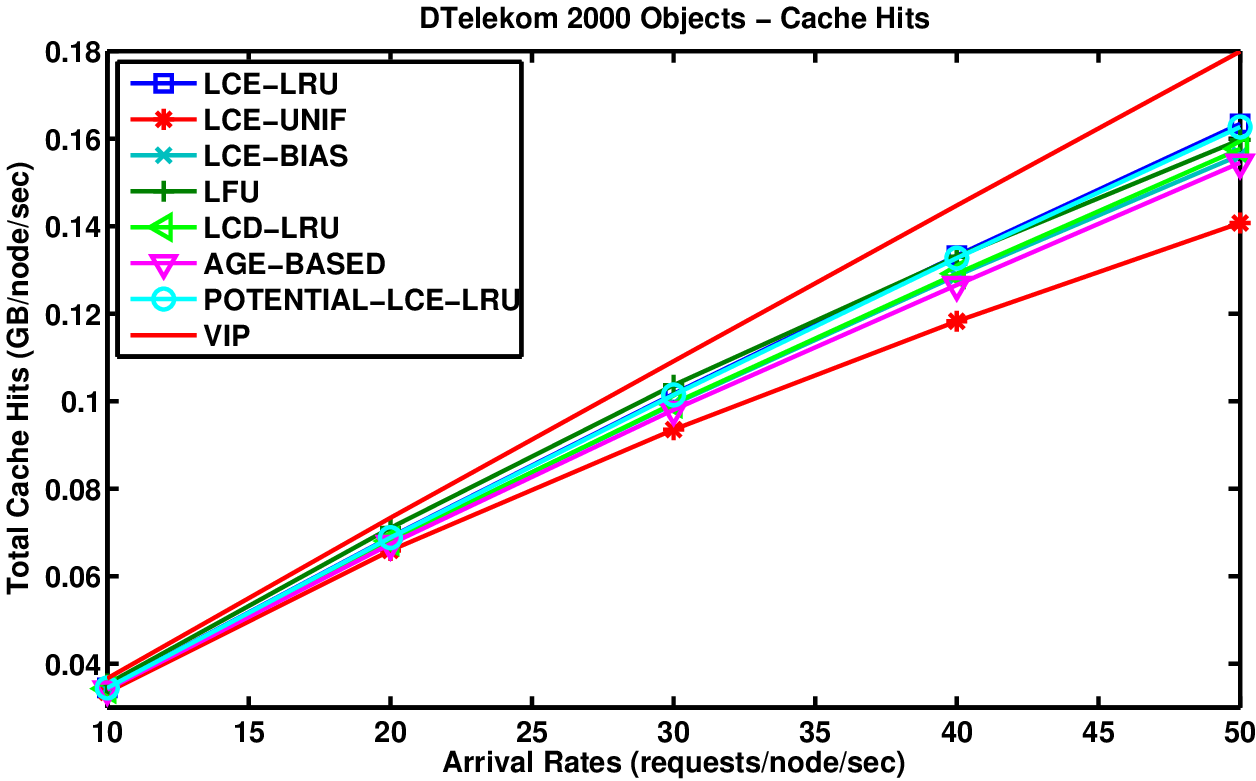}
\caption{DTelekom Topology: Cache Hits}
\vspace{10pt}
\centering
\includegraphics[height=4.8cm]{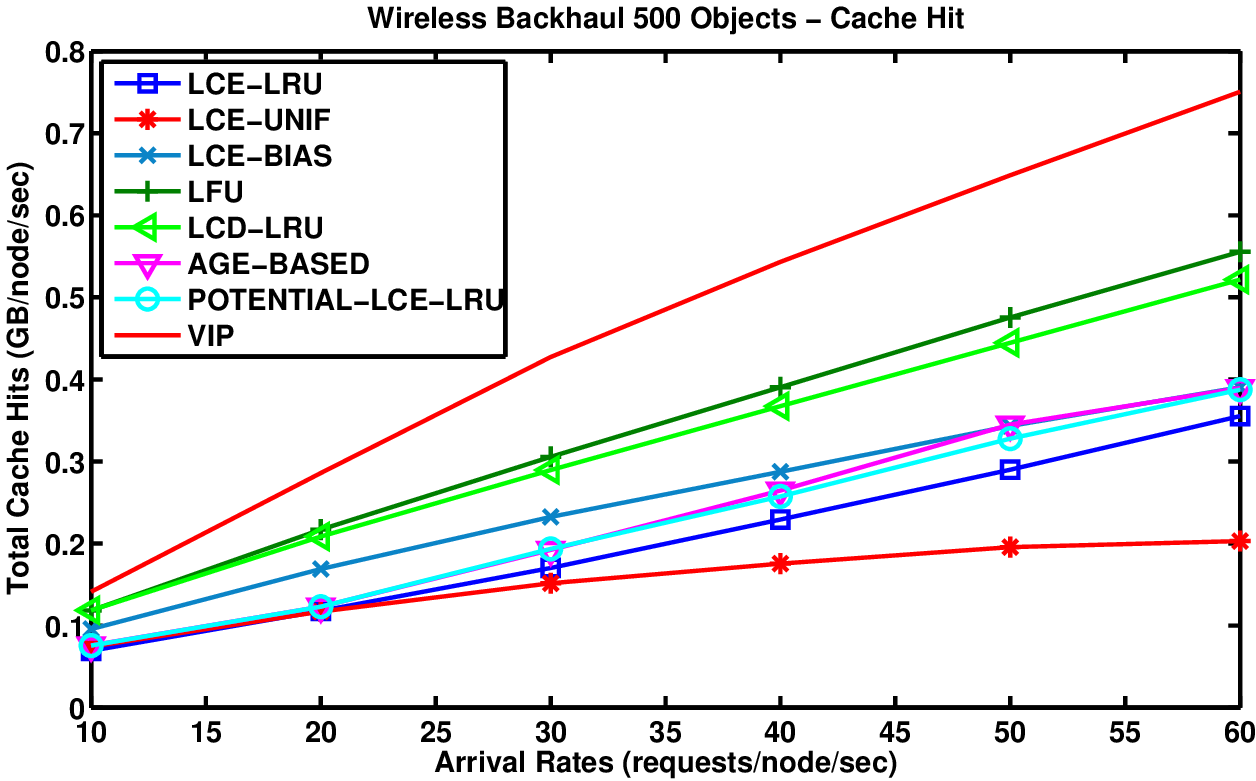}
\caption{Wireless Backhaul Topology: Cache Hits}
\end{center}
\end{figure}

We considered LCE-LRU, LCE-UNIF, and LCE-BIAS combined with shortest path forwarding.   We also considered (under shortest path forwarding) LCD combined with LRU, as well as LCE-LRU combined with potential-based forwarding.
\begin{figure}[t]
\begin{center}
\vspace{-\baselineskip}
\centering
\includegraphics[height=4.8cm]{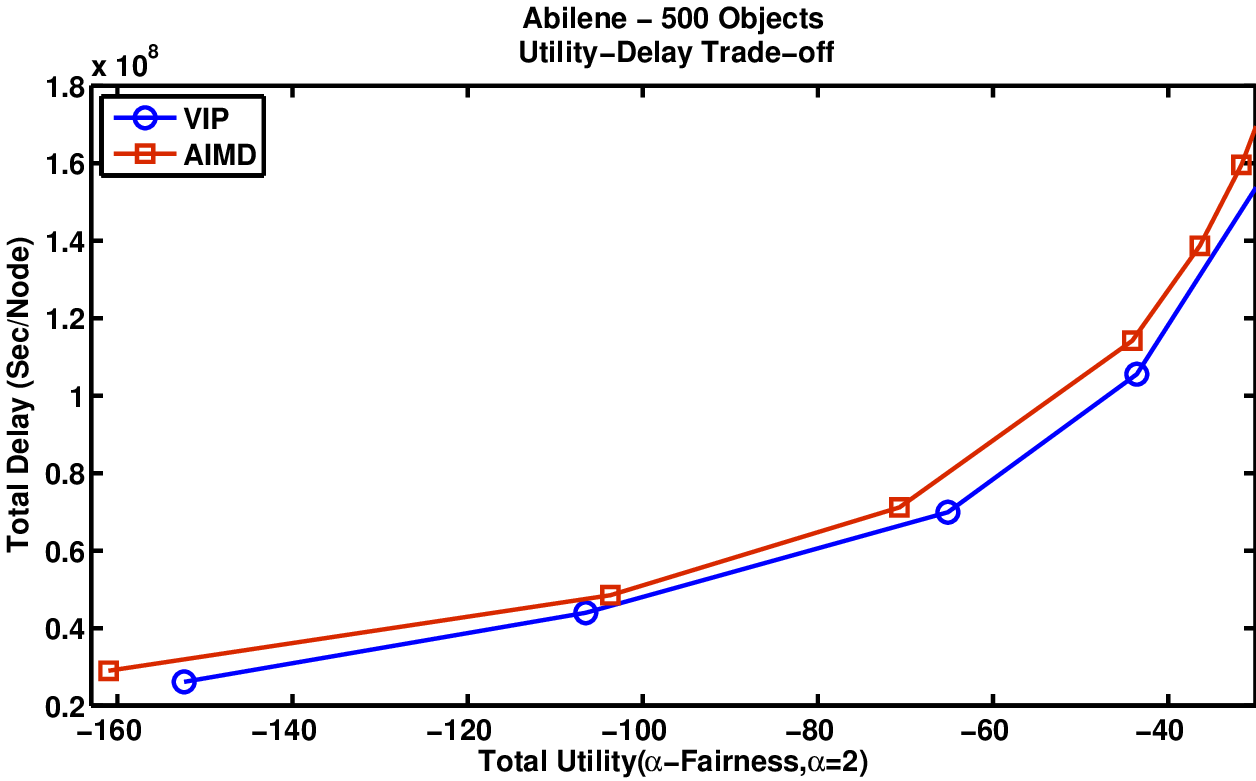}
\caption{Abilene Topology: Utility-Delay Trade-off}

\vspace{10pt}
\centering
\includegraphics[height=4.8cm]{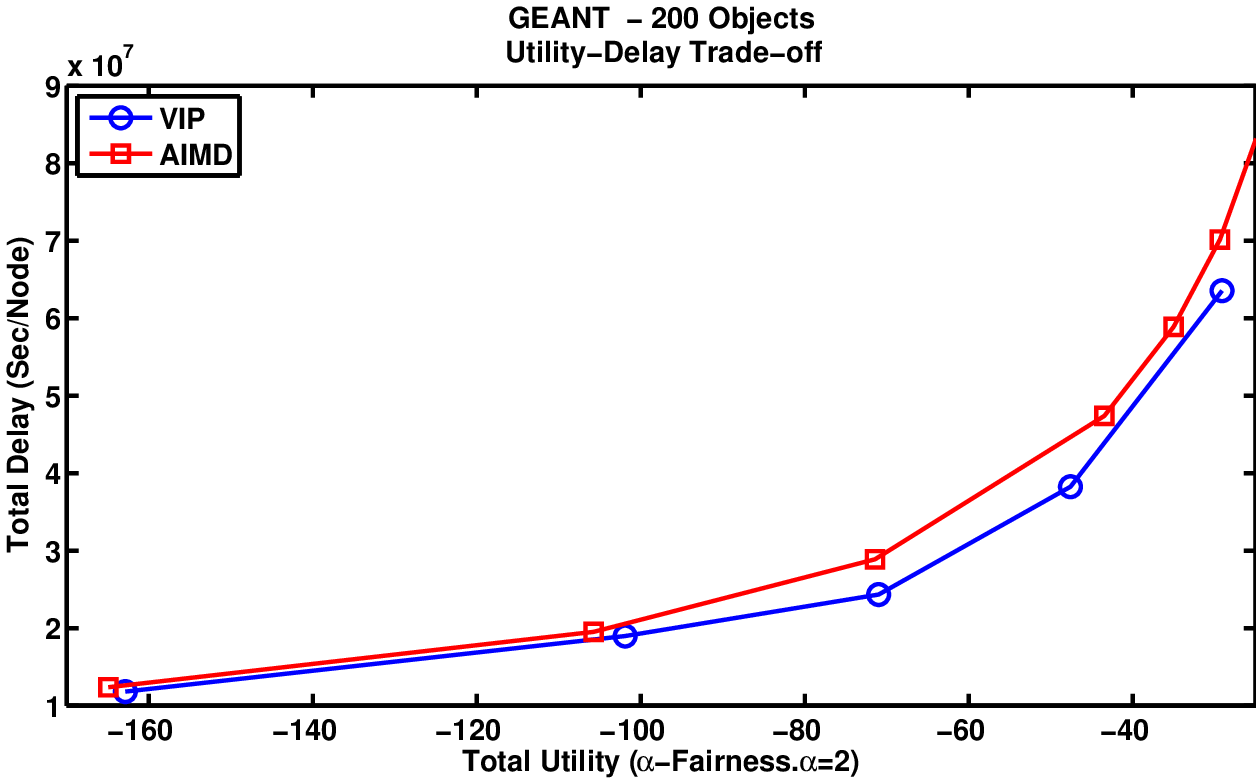}
\caption{GEANT Topology: Utility-Delay Trade-off}

\vspace{10pt}
\centering
\includegraphics[height=4.8cm]{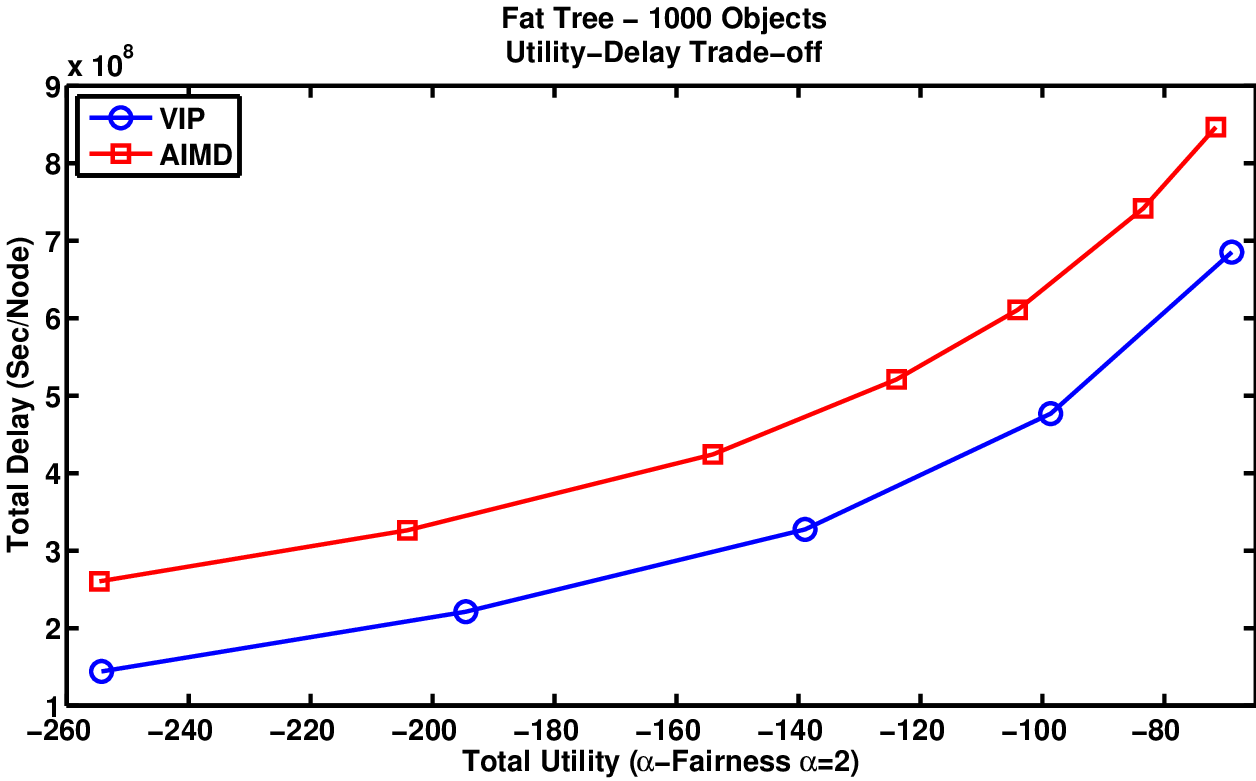}
\caption{Fat Tree Topology: Utility-Delay Trade-off}
\end{center}

\end{figure}

\begin{figure}[t]
\begin{center}
\vspace{-\baselineskip}
\centering
\includegraphics[height=4.8cm]{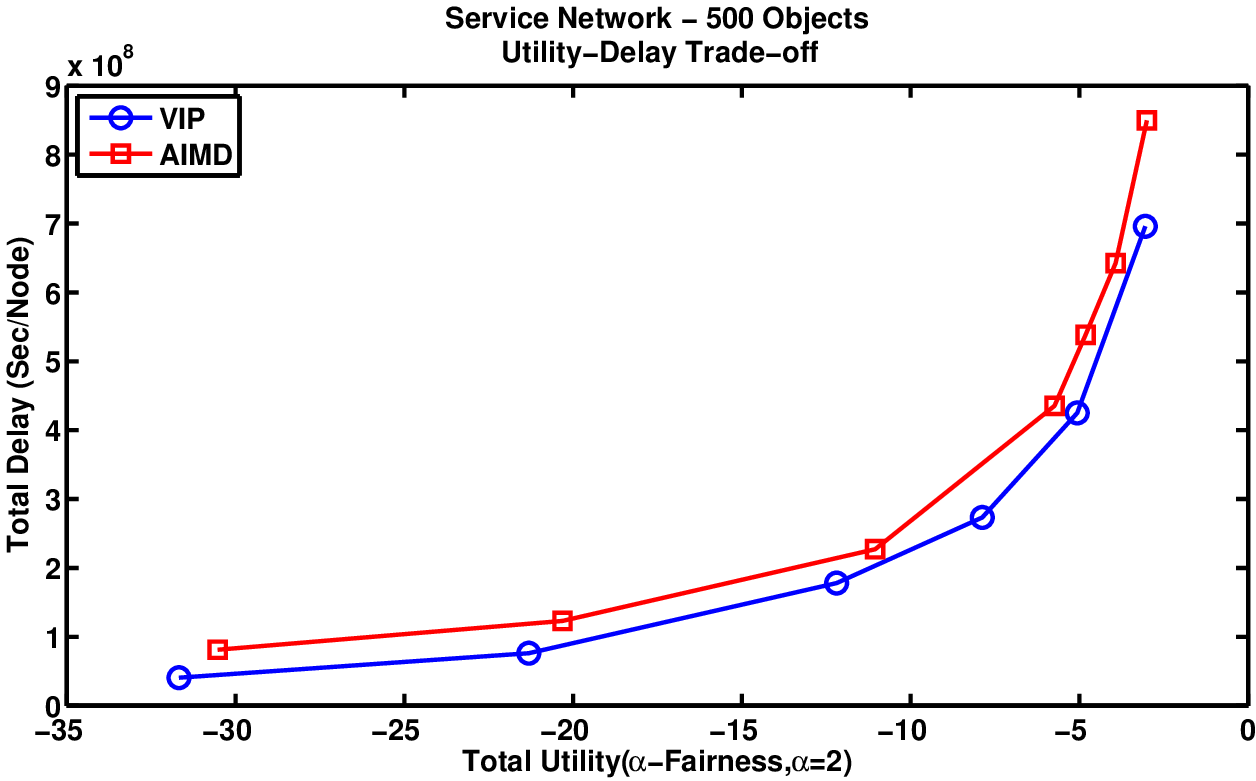}
\caption{Service Network Topology: Utility-Delay Trade-off}
\vspace{10pt}
\centering
\includegraphics[height=4.8cm]{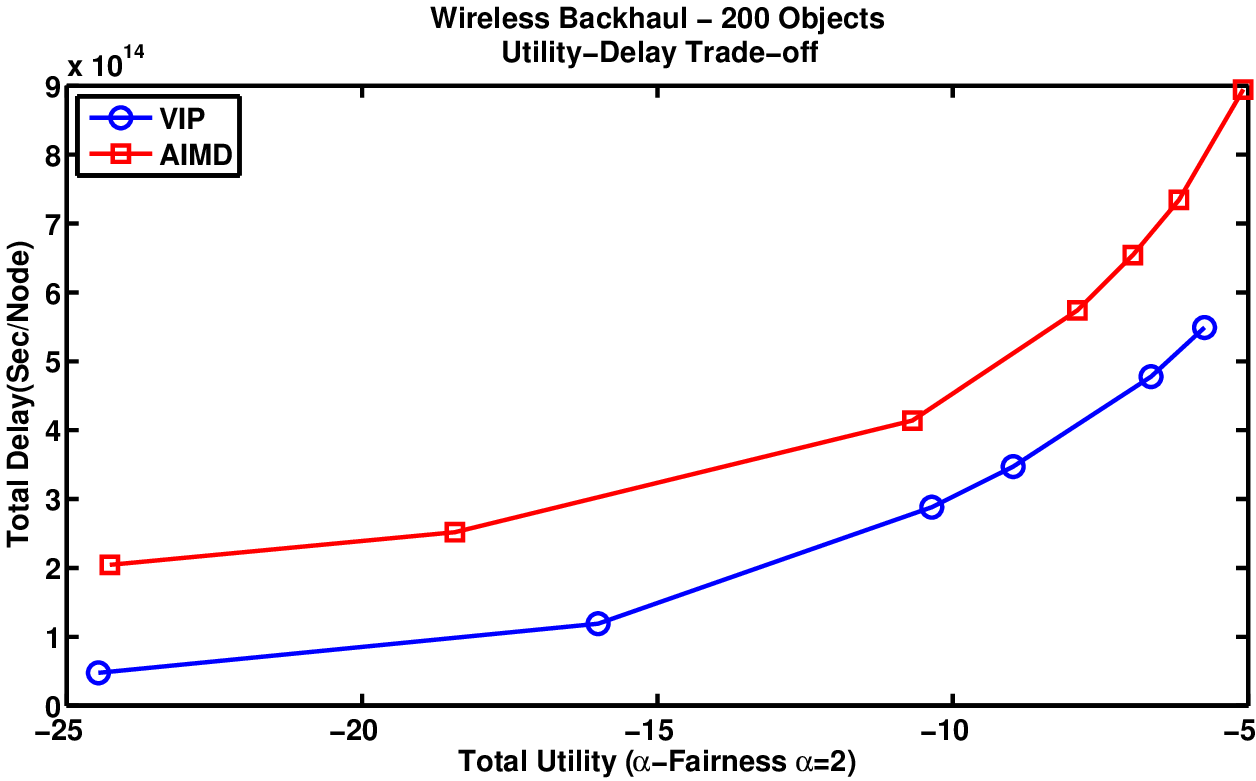}
\caption{Wireless Backhaul Topology: Utility-Delay Trade-off}
\end{center}
\end{figure}

The delay for an Interest Packet request is the difference between the fulfillment
time (i.e., time of arrival of the requested Data Packet) and the creation time
of the Interest Packet request.
A cache hit for a data chunk is recorded when an Interest Packet reaches a node which is not
a content source but which has the data chunk in its  cache.  When
a cache hit occurs, the corresponding metric is incremented by the size
of the chunk in cache.

Figures 3-8 show the delay performance of the algorithms.  It is clear that the Stable Caching VIP Algorithm significantly outperforms all other algorithms tested.  For instance, for the Abilene topology at $\lambda = 100$ requests/node/sec, the total delay for the VIP algorithm is only $65\%$ of the delay for the closest competitor (LCE-LRU), and only about $45\%$ of the delay for the worst performing algorithm (LCE-UNIF).

Figures 9-14 show the cache hit performance  for the algorithms.  Again, the Stable Caching VIP Algorithm  has significantly higher total cache hits than other algorithms.  For the  Fat Tree topology at $\lambda = 60$ requests/node/sec, the total number of cache hits for Algorithm 2 is about $45\%$ higher than that for the closest competitor (LFU) and is almost six times the number of cache hits for the worst performing algorithm (LCE-UNIF).


Note that age-based caching requires prior knowledge of the data object popularity distribution $\{ p_k\}$, while the potential-based routing algorithm needs to inform all network nodes whenever any node changes its cached contents.  Neither of these is required for the VIP algorithm, which measures the demand for various content objects, and enforces distributed and dynamic cache selection and cache replacement.

In sum, the Stable Caching VIP Algorithm significantly outperforms all competing algorithms tested, in terms of user delay and rate of cache hits.

\subsection{Experimental Evaluation of VIP with Congestion Control}
In this section, we present experimental results for the VIP Algorithm with congestion control.
Experiments are carried on the Abilene topology, the GEANT topology, the Service Network topology, the Fat Tree topology and the Wireless Backhaul topology.
When requests are generated at a node, the VIP Congestion Control algorithm controls the number of requests admitted into the network.

Simulation experiments were carried out to compare the utility-delay performance of the Stable Caching VIP algorithm with congestion control
(Algorithm 2 with congestion control) against that of the window-based AIMD (Additive Increase Multiplicative Decrease) congestion control algorithm (combined with a PIT-based forwarding scheme and Least Recently Used (LRU) caching)
described in~\cite{Giovanna}.  In the latter, AIMD window-based congestion control is used at the edges nodes.  A window is kept for each data object. When an edge node successfully fulfills a request for a chunk, the window size for the corresponding data object is increased by $1$.  In the meantime, the round-trip time for the current request is measured and compared to the most recently round-trip times for this data object.  The larger the round-trip time for the current request is relative to the most recent round-trip times, the more likely it is that the window for this data object is decreased by a factor $\beta=0.5$.  For the forwarding part of this algorithm, the node chooses the output interface with the least number of PIT entries for the corresponding data object to forward the Interest Packet.   The $\alpha$-fair utility function~\cite{MoWalrand} with $\alpha=2$ is used for $g^k_n(\cdot)$.  The total utility is the sum of $g^k_n(\cdot)$ over all $k$ and $n$.

Each simulation stops admitting requests upon the expiry of the simulation duration, and terminates when all admitted requests are fulfilled.
We record the number of requests for each object admitted into the network within this period, and calculate the corresponding total utility.  In the mean time, we also calculate the total delay to obtain utility-delay tradeoff curves.

%
%

Figures 15-19 show the utility-delay trade-off for the Congestion Controlled VIP algorithm and the AIMD Window-Based congestion control algorithm.  The VIP algorithm achieves the same total utility with much smaller total delay in all topologies tested.  In sum, the VIP algorithm with or without congestion control significantly improves network performance for a number of different metrics, across a wide range of network topologies.

\section{Conclusion}

The joint design of traffic engineering and caching strategies is central to information-centric architectures such as NDN, which seek to optimally utilize both bandwidth and storage for efficient content distribution.
In this work, we have introduced the VIP framework for the design of high performing NDN networks.  In the virtual plane of the VIP framework, distributed control algorithms operating on virtual interest packets (VIPs) are developed to maximize user demand rate satisfied by the network.  The flow rates and queue lengths of the VIPs are then used to specify the forwarding and caching algorithms in the actual plane, where Interest Packets and Data Packets are processed.   Within the VIP framework, we have developed joint forwarding, caching, and congestion control algorithms which exhibit superior performance in terms of user utility, user delay, and cache hit rates, relative to baseline alternatives.  These algorithms are distributed, dynamic, and do not require prior knowledge of the content popularity distribution.

\bibliography{IEEEabrv,cuiying-bib}
\bibliographystyle{ieeetr}

\section*{Appendix A: Proof of Theorem~\ref{thm:stability}}

The proof of Theorem \ref{thm:stability} involves showing that $\boldsymbol\lambda\in\Lambda$ is necessary for stability and that $\boldsymbol \lambda\in\text{int}(\Lambda)$ is sufficient for stability.
First, we show $\boldsymbol \lambda\in\Lambda$ is necessary for stability.  Suppose the network under arrival rate  $\boldsymbol \lambda$ is stabilizable by some feasible forwarding and caching policy. Let $F_{ab}^k(t)$ denote the number of VIPs for object $k$ transmitted over link $(a,b)$ during slot $t$, satisfying
\begin{align}
&F_{ab}^k(t)\geq 0, \ F_{nn}^k(t)=0, \ F_{src(k)n}^k(t)=0, \quad\nonumber\\
&\hspace{40mm} \forall a,b, n\in \mathcal N, k\in \mathcal K\label{eqn:proof-stability-region-F1}\\
&F_{ab}^k(t)=0, \quad \forall a,b\in \mathcal N, k\in \mathcal K, (a,b)\not \in \mathcal L^k\label{eqn:proof-stability-region-F2}\\
&\sum_{k\in \mathcal K} F_{ab}^k(t)\leq C_{ba}/z, \quad \forall (a,b)\in \mathcal L\label{eqn:proof-capacity-const}
\end{align}
For any slot $\tilde t$, we can define $f_{ab}^k=\sum_{\tau=1}^{\tilde t}F_{ab}^k(\tau)/\tilde t$.
Thus, by \eqref{eqn:proof-stability-region-F1}, \eqref{eqn:proof-stability-region-F2}, and \eqref{eqn:proof-capacity-const}, we can prove \eqref{eqn:stability-region-f1}, \eqref{eqn:stability-region-f2}, and \eqref{eqn:stability-region-capacity}, separately.
Let $S_n^k(t)$ denote the caching state of object $k$ at node $n$ during slot $t$, which satisfies
\begin{align}
&S_n^k(t)\in \{0,1\}, \quad \forall n\in \mathcal N, k\in \mathcal K\label{eqn:stability-region-S}
\end{align}
Define \footnote{Note that $\mathcal T_{n,i,l}\cap \mathcal T_{n,j,m}=\emptyset$ for all $(i,l)\neq (j,m)$ for all $n\in \mathcal N$.}
\begin{align}
\mathcal T_{n,i,l}=\bigg\{\tau\in \{1,\cdots, \tilde t\}:&S_n^k(\tau)=1 \ \forall k\in \mathcal B_{n,i,l}, \nonumber\\ &S_n^k(\tau)=0  \ \forall k\not\in \mathcal B_{n,i,l} \bigg\}\nonumber
\end{align}
for $i=1,\cdots, {K \choose l}$ and $l=0,\cdots, i_n$. Define $\beta_{n,i,l}=T_{n,i,l}/\tilde t$, where  $T_{n,i,l}=|\mathcal T_{n,i,l}|$. Thus, we can prove \eqref{eqn:stability-region-beta} and  \eqref{eqn:stability-region-cache}. It remains to prove \eqref{eq:sink}.
By Lemma 1 of \cite{Neely-Modiano-Rohrs:2005}, network stability implies there exists a finite $M$ such that $V_n^k(t)\leq M$ for all $n\in \mathcal N$ and $k\in \mathcal K$ holds infinitely often. Given an arbitrarily small value $\epsilon >0$, there exists a slot $\tilde t$ such that
\begin{align}
V_n^k(\tilde t)\leq M,\quad 
\frac{M}{\tilde t}\leq \epsilon, \quad
\left|\frac{\sum_{\tau=1}^{\tilde t}A_n^k(\tau)}{\tilde t}-\lambda_n^k\right|\leq \epsilon\label{eqn:proof-queue-bound}
\end{align}
In addition, since for all slot $t$, the queue length is equal to the difference between the total VIPs that have arrived and departed as well as drained, assuming $V_n^k(1)=0$, we have
\begin{align}
&\sum_{\tau=1}^t A_{n}^k(\tau)-V_n^k(t)\nonumber\\
\leq&\sum_{\tau=1}^t\sum_{b\in \mathcal N} F_{nb}^k(\tau)-\sum_{\tau=1}^t\sum_{a\in \mathcal N} F_{an}^k(\tau)+r_n\sum_{\tau=1}^tS_n^k(\tau)\label{eqn:proof-queue-const}
\end{align}
Thus,
by \eqref{eqn:proof-queue-bound} and \eqref{eqn:proof-queue-const},  we have
\begin{align}
&\lambda_n^k-\epsilon\leq \frac{1}{\tilde t}\sum_{\tau=1}^{\tilde t} A_{n}^k(\tau)\nonumber\\
\leq&\frac{1}{\tilde t}V_n^k(\tilde t)+\frac{1}{\tilde t}\sum_{\tau=1}^{\tilde t} \sum_{b\in \mathcal N} F_{nb}^k(\tau)-\frac{1}{\tilde t}\sum_{\tau=1}^{\tilde t} \sum_{a\in \mathcal N} F_{an}^k(\tau)\nonumber\\
&+r_n\frac{1}{\tilde t}\sum_{\tau=1}^{\tilde t} S_n^k(\tau)\label{eqn:proof-stability-region-lambda}
\end{align}
Since $ \sum_{\tau=1}^{\tilde t} S_n^k(\tau)=
\sum_{l=0}^{i_n}\sum_{i=1}^{{K \choose l}} T_{n,i,l}\mathbf 1[k\in \mathcal B_{n,i,l}]$, by
\eqref{eqn:proof-stability-region-lambda}, we have
$$\lambda_n^k\leq \sum_{b\in \mathcal N}f^k_{nb}-\sum_{a\in \mathcal N}f^k_{an}+r_n\sum_{l=0}^{i_n}\sum_{i=1}^{{K \choose l}} \beta_{n,i,l}\mathbf 1[k\in \mathcal B_{n,i,l}]+2\epsilon.$$
By letting $\epsilon\to 0$, we can prove \eqref{eq:sink}.

Next, we show $\boldsymbol \lambda\in\text{int}(\Lambda)$ is sufficient for stability. $\boldsymbol \lambda\in\text{int}(\Lambda)$ implies that there exists $\boldsymbol \epsilon=\left(\epsilon_n^k\right)$, where $\epsilon_n^k>0$, such that $\boldsymbol \lambda+\boldsymbol \epsilon\in \Lambda$. Let $\left(f_{ab}^k\right)$ and $\left(\boldsymbol \beta_n\right)$ denote the  flow variables  and storage variables associated with arrival rates $\boldsymbol \lambda+\boldsymbol \epsilon$. Thus, \eqref{eqn:stability-region-f1}, \eqref{eqn:stability-region-f2}, \eqref{eqn:stability-region-beta}, \eqref{eqn:stability-region-capacity}, \eqref{eqn:stability-region-cache}, and
\begin{align}
&\lambda_n^k+\epsilon_n^k\nonumber\\
\leq& \sum_{b\in \mathcal N}f^k_{nb}-\sum_{a\in \mathcal N}f^k_{an}+r_n\sum_{l=0}^{i_n}\sum_{i=1}^{{K \choose l}} \beta_{n,i,l}\mathbf 1[k\in \mathcal B_{n,i,l}],\nonumber\\
&\hspace{30mm}\forall n\in \mathcal N, k\in \mathcal K, n\neq src(k)\label{eqn:stability-region-lambda-epsilon}
\end{align}
hold. We now construct the randomized forwarding policy. For every link $(a,b)$ such that $\sum_{k\in \mathcal K}f_{ab}^k>0$, transmit the VIPs of the single object $\tilde k_{ab}$, where  $\tilde k_{ab}$ is chosen randomly to be $k$ with probability $f_{ab}^k/\sum_{k\in \mathcal K}f_{ab}^k$. Then, the number of VIPs that can be transmitted in slot $t$ is as follows:
\begin{align}
\tilde \mu_{ab}^k(t)=
\begin{cases}
\sum_{k\in \mathcal K}f_{ab}^k, & \text{if  $k=\tilde k_{ab}$} \\
0, & \text{otherwise}
\end{cases}
\end{align}
Null bits are delivered if there are not enough bits in a queue. For every link $(a,b)$ such that $\sum_{k\in \mathcal K}f_{ab}^k=0$, choose $\tilde \mu_{ab}^k(t)=0$ for all $k\in\mathcal K$. Thus, we have
\begin{align}
\mathbb E\left[\tilde \mu_{ab}^k(t)\right]=f_{ab}^k\label{eqn:rand-mu-ave}
\end{align}
Next, we construct the randomized caching policy. For every node $n$, cache the single combination $\tilde{\mathcal B}_n$, where $\tilde{\mathcal B}_n$ is chosen randomly to be $\mathcal B_{n,i,l}$ with probability $\beta_{n,i,l}/\sum_{l=0}^{i_n}\sum_{i=1}^{{K \choose l}} \beta_{n,i,l}=\beta_{n,i,l},$ as $\sum_{l=0}^{i_n}\sum_{i=1}^{{K \choose l}} \beta_{n,i,l}=1$ by \eqref{eqn:stability-region-cache}. Then, the caching state in slot $t$ is as follows:
\begin{align}
\tilde s_n^k(t)=
\begin{cases}
1, & \text{if  $k\in \tilde{\mathcal B}_n$} \\
0, & \text{otherwise}
\end{cases}
\end{align}
Thus, we have
\begin{align}
\mathbb E\left[\tilde s_n^k(t)\right]=\sum_{l=0}^{i_n}\sum_{i=1}^{{K \choose l}} \beta_{n,i,l}\mathbf 1[k\in \mathcal B_{n,i,l}] \label{eqn:rand-s-ave}
\end{align}
Therefore, by \eqref{eqn:rand-mu-ave}, \eqref{eqn:rand-s-ave} and \eqref{eqn:stability-region-lambda-epsilon}, we have
\begin{align}
&\mathbb E\left[\left(\sum_{b\in \mathcal N}\tilde{\mu}^k_{nb}(t)-\sum_{a\in \mathcal N}\tilde{\mu}^k_{an}(t)+r_n^{(k)}\tilde{s}_n^k(t)\right)\right]\nonumber\\
=&\sum_{b\in \mathcal N}f^k_{nb}-\sum_{a\in \mathcal N}f^k_{an}+r_n\sum_{l=0}^{i_n}\sum_{i=1}^{{K \choose l}} \beta_{n,i,l}\mathbf 1[k\in \mathcal B_{n,i,l}]\nonumber\\
\geq &\lambda_n^k+\epsilon_n^k\label{eqn:rand-policy-ineq}
\end{align}
In other words, the arrival rate is less than the service rate.
Thus, by Loynes' theorem\cite{Loynes:62}, we can show that the network is stable.

\section*{Appendix B: Proof of Theorem~\ref{Thm:thpt-opt}}

Define the quadratic Lyapunov function $\mathcal L(\mathbf V)\triangleq
\sum_{n\in \mathcal N,k\in \mathcal K}(V^k_n)^2$. The Lyapunov drift at slot $t$ is given by
$\Delta (\mathbf V(t))\triangleq \mathbb E[\mathcal L\big(\mathbf
V(t+1)\big)-\mathcal L\big(\mathbf V(t)\big)|\mathbf V(t)]$. First, we calculate $\Delta (\mathbf V(t))$.
Taking square on both sides of \eqref{eqn:queue_dyn}, we have
\begin{align}
&\left(V^k_n(t+1)\right)^2\nonumber\\
\leq&\Bigg(\Bigg(
\left(V^k_n(t)-\sum_{b\in \mathcal N}\mu^k_{nb}(t)\right)^+ +A^k_n(t)
\nonumber\\
&+\sum_{a\in \mathcal N}\mu^k_{an}(t)-r_ns_n^k(t)\Bigg)^+\Bigg)^2\nonumber\\
\leq&\Bigg(
\left(V^k_n(t)-\sum_{b\in \mathcal N}\mu^k_{nb}(t)\right)^+ +A^k_n(t)\nonumber\\
&+\sum_{a\in \mathcal N}\mu^k_{an}(t)-r_ns_n^k(t)\Bigg)^2\nonumber\\
\leq&\left(V^k_n(t)-\sum_{b\in \mathcal N}\mu^k_{nb}(t)\right)^2+2
\left(V^k_n(t)-\sum_{b\in \mathcal N}\mu^k_{nb}(t)\right)^+\nonumber\\
&\times\left(A^k_n(t)
+\sum_{a\in \mathcal N}\mu^k_{an}(t)-r_ns_n^k(t)\right)\nonumber\\
&+\left(A^k_n(t)
+\sum_{a\in \mathcal N}\mu^k_{an}(t)-r_ns_n^k(t)\right)^2
\nonumber\\
=&\left(V^k_n(t)\right)^2+\left(\sum_{b\in \mathcal N}\mu^k_{nb}(t)\right)^2-2 V^k_n(t)\sum_{b\in \mathcal N}\mu^k_{nb}(t) \nonumber\\
&+\left(A^k_n(t)+\sum_{a\in \mathcal N}\mu^k_{an}(t)-r_n^ks_n^k(t)\right)^2\nonumber\\
&+2\left(V^k_n(t)-\sum_{b\in \mathcal N}\mu^k_{nb}(t)\right)^+\left(A^k_n(t)
+\sum_{a\in \mathcal N}\mu^k_{an}(t)\right)\nonumber\\
&-2\left(V^k_n(t)-\sum_{b\in \mathcal N}\mu^k_{nb}(t)\right)^+r_n^ks_n^k(t)\nonumber\\
\leq&\left(V^k_n(t)\right)^2+\left(\sum_{b\in \mathcal N}\mu^k_{nb}(t)\right)^2-2 V^k_n(t)\sum_{b\in \mathcal N}\mu^k_{nb}(t) \nonumber\\
&+\left(A^k_n(t)+\sum_{a\in \mathcal N}\mu^k_{an}(t)+r_n^ks_n^k(t)\right)^2\nonumber\\
&+2V^k_n(t)\left(A^k_n(t)
+\sum_{a\in \mathcal N}\mu^k_{an}(t)\right)\nonumber\\
&-2\left(V^k_n(t)-\sum_{b\in \mathcal N}\mu^k_{nb}(t)\right)r_n^ks_n^k(t)\nonumber\\
\leq&\left(V^k_n(t)\right)^2+\left(\sum_{b\in \mathcal N}\mu^k_{nb}(t)\right)^2+2\sum_{b\in \mathcal N}\mu^k_{nb}(t)r_ns_n^k(t)\nonumber\\
&+\left(A^k_n(t)+\sum_{a\in \mathcal N}\mu^k_{an}(t)+r_ns_n^k(t)\right)^2\nonumber\\
&+2V^k_n(t)A^k_n(t)-2V^k_n(t)\left(\sum_{b\in \mathcal N}\mu^k_{nb}(t)-\sum_{a\in \mathcal N}\mu^k_{an}(t)\right)\nonumber\\
&-2V^k_n(t)r_ns_n^k(t)\nonumber
\end{align}
Summing over all $n,k$, we have
\begin{align}
&\mathcal L\left(\mathbf V(t+1)\right)-\mathcal L\left(\mathbf V(t)\right)\nonumber\\
\stackrel{(a)}{\leq}&2N B+2\sum_{n\in \mathcal N,k\in \mathcal K}V^k_n(t)A^k_n(t)\nonumber\\
&-2\sum_{(a,b)\in \mathcal L}\sum_{k\in \mathcal K}
\mu^k_{ab}(t)\big(V^k_a(t)-V^k_b(t)\big)\nonumber\\
&-2\sum_{n\in \mathcal N,k\in \mathcal K}V^k_n(t)r_ns_n^k(t)\label{eqn:proof-deltaL}
\end{align}
where (a) is due to the following:
\begin{align}
&\sum_{k\in \mathcal K}\left(\sum_{b\in \mathcal N}\mu^k_{nb}(t)\right)^2\leq \left(\sum_{k\in \mathcal K}\sum_{b\in \mathcal N}\mu^k_{nb}(t)\right)^2\leq\left(\mu^{out}_{n, \max}\right)^2,\nonumber\\
& \sum_{k\in \mathcal K}\left(A^k_n(t)+\sum_{a\in \mathcal N}\mu^k_{an}(t)+r_ns_n^k(t)\right)^2\nonumber\\
&\leq \left(\sum_{k\in \mathcal K}\left(A^k_n(t)+\sum_{a\in \mathcal N}\mu^k_{an}(t)+r_ns_n^k(t)\right)\right)^2\nonumber\\
 &\leq (A_{n,\max}+\mu^{in}_{n,\max}+r_{n,\max})^2, \nonumber\\
&\sum_{k\in \mathcal K}\sum_{b\in \mathcal N}\mu^k_{nb}(t)r_ns_n^k(t)\nonumber\\
&\leq \left(\sum_{k\in \mathcal K}\sum_{b\in \mathcal N}\mu^k_{nb}(t)\right)\left(\sum_{k\in \mathcal K}r_ns_n^k(t)\right)\leq \mu^{out}_{n,
\max}r_{n,\max},\nonumber\\
&\sum_{n\in \mathcal N,k\in \mathcal K}V^k_n(t)\left(\sum_{b\in\mathcal N}\mu^k_{nb}(t)-\sum_{a\in \mathcal N}\mu^k_{an}(t)\right)\nonumber\\
&=\sum_{(a,b)\in \mathcal L}\sum_{k\in \mathcal K}
\mu^k_{ab}(t)\big(V^k_a(t)-V^k_b(t)\big).\nonumber
\end{align}
Taking conditional
expectations on both sides of \eqref{eqn:proof-deltaL}, we have
\begin{align}
&\Delta (\mathbf
V(t))\nonumber\\
\leq&2N B+2\sum_{n\in \mathcal N,k\in \mathcal K}V^k_n(t)\lambda^k_n\nonumber\\
&-2\mathbb
E\left[\sum_{(a,b)\in \mathcal L}\sum_{k\in \mathcal K}
\mu^k_{ab}(t)\left(V^k_a(t)-V^{(c)}_b(t)\right)|\mathbf V(t)\right]\nonumber\\
&-2\mathbb E\left[\sum_{n\in \mathcal N,k\in \mathcal K}V^k_n(t)r_n s_n^k(t)|\mathbf V(t)\right]\nonumber\\
\stackrel{(b)}{\leq}&2N B+2\sum_{n\in \mathcal N,k\in \mathcal K}V^k_n(t)\lambda^k_n\nonumber\\
&-2\mathbb
E\left[\sum_{(a,b)\in \mathcal L}\sum_{k\in \mathcal K}
\tilde{\mu}^k_{ab}(t)\left(V^k_a(t)-V^k_b(t)\right)|\mathbf V(t)\right]\nonumber\\
&-2\mathbb E\left[\sum_{n\in \mathcal N,k\in \mathcal K}V^k_n(t)r_n\tilde{s}_n^k(t)|\mathbf V(t)\right]\nonumber\\
=&2N B+2\sum_{n\in \mathcal N,k\in \mathcal K}V^k_n(t)\lambda^k_n-2\sum_{n\in \mathcal N,k\in \mathcal K}V^k_n(t)\nonumber\\
&\times\mathbb E\left[\left(\sum_{b\in \mathcal N}\tilde{\mu}^k_{nb}(t)-\sum_{a\in \mathcal N}\tilde{\mu}^k_{an}(t)+r_n\tilde{s}_n^k(t)\right)|\mathbf
V(t)\right]\label{eqn:proof_ineq0}
\end{align}
where (b) is due to the fact that Algorithm \ref{Alg:VIP} minimizes the
R.H.S. of  (b) over all feasible $\tilde{\mu}^k_{ab}(t)$ and $\tilde{s}_n^k(t)$.\footnote{ Note that $\mu^k_{ab}(t)$ and $s_n^k(t)$ denote the actions of Algorithm \ref{Alg:VIP}.}  Since $\boldsymbol
\lambda+\boldsymbol \epsilon\in \Lambda$, according to the proof of Theorem \ref{thm:stability}, there exists a stationary
randomized forwarding and caching policy that makes decisions
independent of $\mathbf V(t)$ such that
\begin{align}
&\mathbb E\left[\left(\sum_{b\in \mathcal N}\tilde{\mu}^k_{nb}(t)-\sum_{a\in \mathcal N}\tilde{\mu}^k_{an}(t)+r_n\tilde{s}_n^k(t)\right)|\mathbf
V(t)\right]\nonumber\\
\geq &\lambda^k_n+\epsilon^k_n\label{eqn:proof_ineq1}
\end{align}
Substituting \eqref{eqn:proof_ineq1}
into \eqref{eqn:proof_ineq0}, we have $\Delta (\mathbf V(t))
\leq 2N B-2\sum_{n\in \mathcal N,k\in \mathcal K} \epsilon^k_nV^k_n(t)
\leq
2N B-2\epsilon\sum_{n\in \mathcal N,k\in \mathcal K}V^k_n(t) $.
By  Lemma 4.1 of \cite{Georgiadis-Neely-Tassiulas:2006}, we complete the proof.

\section*{Appendix C: Proof of Theorem~\ref{Thm:flow-control}}

Define the Lyapunov function $\mathcal L(\boldsymbol \Theta)\triangleq
\sum_{n\in \mathcal N,k\in \mathcal K}\left((V^k_n)^2+(Y^k_n)^2\right)$, where $\boldsymbol \Theta\triangleq (\mathbf V, \mathbf Y)$. The Lyapunov drift at slot $t$ is
$\Delta (\boldsymbol \Theta(t))\triangleq \mathbb E[\mathcal L\big(\boldsymbol \Theta(t+1)\big)- \mathcal  L\left(\boldsymbol \Theta(t)\right)|\boldsymbol \Theta(t)]$.  First, we calculate $\Delta (\boldsymbol \Theta(t))$. Similar to Appendix B, taking square on both sides of \eqref{eqn:queue_dyn-flow}, we have
\begin{align}
&\left(V^k_n(t+1)\right)^2\nonumber\\
\leq&\left(V^k_n(t)\right)^2+\left(\sum_{b\in \mathcal N}\mu^k_{nb}(t)\right)^2+2\sum_{b\in \mathcal N}\mu^k_{nb}(t)r_ns_n^k(t)\nonumber\\
&+\left(\alpha^k_n(t)+\sum_{a\in \mathcal N}\mu^k_{an}(t)+r_ns_n^k(t)\right)^2\nonumber\\
&+2V^k_n(t)\alpha^k_n(t)-2V^k_n(t)\left(\sum_{b\in \mathcal N}\mu^k_{nb}(t)-\sum_{a\in \mathcal N}\mu^k_{an}(t)\right)\nonumber\\
&-2V^k_n(t)r_ns_n^k(t)\nonumber
\end{align}
In addition, taking square on both sides of \eqref{eqn:queue_dyn-virtual}, we have
\begin{align}
&\left(Y^k_n(t+1)\right)^2\nonumber\\
\leq&\left(Y^k_n(t)\right)^2+\left(\alpha^k_n(t)\right)^2+\left(\gamma^k_n(t)\right)^2 -2Y_n^k(t)\left(\alpha_n^k(t)-\gamma_n^k(t)\right)\nonumber
\end{align}
Therefore, we have
\begin{align}
&\mathcal  L\left(\boldsymbol \Theta (t+1)\right)- \mathcal  L\left(\boldsymbol \Theta (t)\right)\nonumber\\
 \leq&
2N\hat B+2\sum_{n\in \mathcal N,k\in \mathcal K}V^k_n(t)\alpha^k_n(t)\nonumber\\
&  -2\sum_{n\in \mathcal N,k\in \mathcal K}V^k_n(t)\left(\sum_{b\in \mathcal N}\mu^k_{nb}(t)-\sum_{a\in \mathcal N}\mu^k_{an}(t)\right)\nonumber\\
& -2\sum_{n\in \mathcal N,k\in \mathcal K}V^k_n(t)r_ns_n^k(t)\nonumber\\
& -2\sum_{n\in \mathcal N,k\in \mathcal K}Y_n^k(t)\left(\alpha_n^k(t)-\gamma_n^k(t)\right)
\label{eqn:proof-drift-flow}
\end{align}
Taking conditional
expectations and subtracting $$W\mathbb E\left[\sum_{n\in \mathcal N,k\in \mathcal K}g^k_n\left( \gamma^k_n(t)\right)|\boldsymbol \Theta (t)\right]$$ from both sides of \eqref{eqn:proof-drift-flow}, we have
\begin{align}
&\Delta \left(\boldsymbol \Theta (t)\right )-W\mathbb E\left[\sum_{n\in \mathcal N,k\in \mathcal K}g^k_n\left( \gamma^k_n(t)\right)|\boldsymbol \Theta(t)\right]\nonumber\\
 \stackrel{(a)}{\leq}& 2N\hat B-2\sum_{n\in \mathcal N,k\in \mathcal K}\left(Y^k_n(t)-V^k_n(t)\right)\mathbb E\left[\tilde \alpha^k_n(t)|\boldsymbol \Theta(t)\right]\nonumber\\
& -\sum_{n\in \mathcal N,k\in \mathcal K}\mathbb E\left[Wg^k_n\left( \tilde \gamma^k_n(t)\right)-2Y_n^k(t)\tilde \gamma^k_n(t)|\boldsymbol \Theta(t)\right]\nonumber\\
&  - 2\sum_{n\in \mathcal N,k\in \mathcal K}V^k_n(t)\nonumber\\
&\times\mathbb E\left[\left(\sum_{b\in \mathcal N}\tilde{\mu}^k_{nb}(t)-\sum_{a\in \mathcal N}\tilde{\mu}^k_{an}(t)+r_n\tilde{s}_n^k(t)\right)|\boldsymbol \Theta(t)\right]
\label{eqn:proof_ineq0-flow}
\end{align}
where (a) is due to the fact that  Algorithm \ref{Alg:flow-DBP} minimizes the
R.H.S. of  (b) over all possible alternative
$\tilde \alpha^k_n (t)$, $\tilde \gamma^k_n(t)$, $\tilde \mu^k_{ab}(t)$, and $\tilde s^k_n(t)$.\footnote{Note that $\alpha^k_n (t)$, $\gamma^k_n(t)$, $\mu^k_{ab}(t)$ and $s^k_n(t)$
denote the actions of Algorithm \ref{Alg:flow-DBP}.}
It is not difficult to construct alternative random policies that choose $\tilde \alpha^k_n (t)$, $\tilde \gamma^k_n(t)$, $\tilde \mu^k_{ab}(t)$ and $\tilde s^k_n(t)$ such that
\begin{align}
&\mathbb E\left[\tilde \alpha^k_n(t)|\boldsymbol \Theta(t)\right]=\overline \alpha^{k*}_n(\boldsymbol \epsilon)\label{eqn:rand-r}\\
& \tilde \gamma^k_n(t)=\overline \alpha^{k*}_n(\boldsymbol \epsilon)\label{eqn:rand-gamma}\\
& \mathbb E\left[\left(\sum_{b\in \mathcal N}\tilde{\mu}^k_{nb}(t)-\sum_{a\in \mathcal N}\tilde{\mu}^k_{an}(t)+r_n\tilde{s}_n^k(t)\right)|\boldsymbol \Theta(t)\right]\nonumber\\
 \geq& \overline \alpha^{k*}_n(\boldsymbol \epsilon)+\epsilon^k_n \label{eqn:rand-mu}
\end{align}
where  $ \overline{\boldsymbol \alpha}^*(\boldsymbol \epsilon  )=(\overline \alpha^{k*}_n(\boldsymbol \epsilon))$ is the target $\boldsymbol \epsilon $-optimal admitted rate given by
\eqref{eqn:eps-opt-prob}.\footnote{Specifically, \eqref{eqn:rand-r} can be achieved by the random policy setting $\tilde \alpha^k_n(t)=A^k_n(t)$ with probability $\overline \alpha^{k*}_n(\boldsymbol \epsilon)/\lambda^k_n$ and $\tilde \alpha^k_n(t)=0$ with probability $1-\overline \alpha^{k*}_n(\boldsymbol \epsilon)/\lambda^k_n$.}
\eqref{eqn:rand-mu} follows from the same arguments leading to~\eqref{eqn:proof_ineq1}.
Thus, by \eqref{eqn:rand-r}, \eqref{eqn:rand-gamma} and \eqref{eqn:rand-mu}, from \eqref{eqn:proof_ineq0-flow}, we obtain
\begin{align}
&\Delta (\boldsymbol \Theta (t))-W\mathbb E\left[\sum_{n\in \mathcal N,k\in \mathcal K}g^k_n\left( \gamma^k_n(t)\right)|\boldsymbol \Theta(t)\right]\nonumber\\
 \leq &  2N\hat B-2\min_{n\in \mathcal N,k\in \mathcal K}\left\{\epsilon_n^k\right\}  \sum_{n\in \mathcal N,k\in \mathcal K}V_n^k(t)\nonumber\\
 &-W\sum_{n\in \mathcal N,k\in \mathcal K}g^k_n\left(\overline \alpha^{k*}_n(\boldsymbol \epsilon )\right)\nonumber
\end{align}
Applying Theorem 5.4 of \cite{Georgiadis-Neely-Tassiulas:2006}, we have
\begin{align}
&\limsup_{t\to\infty}\frac{1}{t}\sum_{\tau=1}^t\sum_{n\in \mathcal N,k\in \mathcal K} \mathbb
E[V^k_n(\tau)]  \nonumber\\
\leq& \frac{2N\hat B+WG_{\max}}{
2\min_{n\in \mathcal N,k\in \mathcal K}\left\{\epsilon_n^k \right\}}\label{eqn:proof-enhanced-flow-DBP-U}\\
&\liminf_{t\to\infty}\sum_{n\in \mathcal N,k\in \mathcal K} g^k_n\left(\overline \gamma^k_n(t)\right)\nonumber\\
  \geq&
\sum_{n\in \mathcal N,k\in \mathcal K}
g^k_n\left( \overline \alpha^{k*}_n\left(\boldsymbol \epsilon\right)\right)-\frac{2N\hat B}{W}\label{eqn:proof-enhanced-flow-DBP-g}
\end{align}
As in \cite[page 88]{Georgiadis-Neely-Tassiulas:2006}, we optimize the R.H.S. of  \eqref{eqn:proof-enhanced-flow-DBP-U} and \eqref{eqn:proof-enhanced-flow-DBP-g} over all possible $\boldsymbol \epsilon\in \Lambda$. Thus, we can show \eqref{eqn:enhanced-flow-DBP-U} and
\begin{align}
& \liminf_{t\to\infty}\sum_{n\in \mathcal N,k\in \mathcal K}g^k_n\left(\overline \gamma^k_n(t)\right) \nonumber\\
\geq&
\sum_{n\in \mathcal N,k\in \mathcal K}
g^k_n\left(\overline \alpha^{k*}_n\left(\mathbf 0\right)\right)-\frac{2N\hat B}{W}
\label{eqn:enhanced-flow-DBP-g-sim}
\end{align}
where $\overline \gamma^k_n(t)\triangleq \frac{1}{t}\sum_{\tau=1}^t\mathbb E[ \gamma^k_n(\tau)]$.
It is easy to prove $\overline \gamma^k_n(t)\leq \overline \alpha^k_n(t)$ by showing the stability of the virtual queues. Thus, we can show \eqref{eqn:enhanced-flow-DBP-g} based on \eqref{eqn:enhanced-flow-DBP-g-sim}. We complete the proof.

\end{document}